\documentclass[journal]{IEEEtran}

\usepackage{amsmath,amssymb,alltt,times,mathptmx,graphicx}
\usepackage{courier,subfigure,color,url}
\usepackage{xspace}
\usepackage{cite}
\usepackage{amsfonts}

\newtheorem{definition}{Definition}
\newtheorem{lemma}{Lemma}
\newtheorem{theorem}{Theorem}

\newcommand{\etal}{\emph{et al.}\xspace}

\newcommand{\comment}[1]{}

\begin{document}
%
\title{Networked Computing in Wireless Sensor Networks for Structural Health Monitoring}
%
%
%

\author{Apoorva Jindal,~\IEEEmembership{Member,~IEEE,}
        Mingyan Liu,~\IEEEmembership{Member,~IEEE,}
\thanks{A. Jindal and M. Liu are with the Department of Electrical Engineering and Computer Science at University of Michigan, Ann Arbor.
E-mail: apoorvaj@umich.edu,mingyan@eecs.umich.edu.}
}


%

\maketitle

\begin{abstract}
This paper studies the problem of distributed computation over a network of wireless sensors. 
While this problem applies to many emerging applications, 
to keep our discussion concrete we will focus on sensor networks used for structural health monitoring.  
Within this context, the heaviest computation is to determine the singular value decomposition (SVD) to extract mode shapes (eigenvectors) of a structure.  
Compared to collecting raw vibration data and performing SVD at a central location, computing SVD within the network can result in 
significantly lower energy consumption and delay.
Using recent results on decomposing SVD, a well-known centralized operation, 
into components, we seek to determine a near-optimal
communication structure that enables the distribution of this computation and the reassembly of the final results, 
with the objective of minimizing energy consumption subject to a computational delay constraint.  
We show that this reduces to a generalized clustering problem; a cluster forms a unit on which a component of the overall computation is performed.  
We establish that this problem is NP-hard.  By relaxing the delay constraint, we derive a lower bound to this problem.  
We then propose an integer linear program (ILP) to solve the constrained problem exactly as well as an approximate algorithm with a proven approximation ratio.  We further present a distributed version of the approximate algorithm.  We present both simulation and experimentation results to demonstrate the effectiveness of these algorithms.
\end{abstract}


\IEEEpeerreviewmaketitle

\section{Introduction}
\label{sec:intro}

Over the past decade, tremendous progress has been made in understanding and using wireless sensor networks.  Of particular relevance to this paper are extensive studies on in-network processing, e.g., finding efficient routing strategies when data compression and aggregation are involved. However, many emerging applications, e.g., body area sensing, structural health monitoring, and various other cyber-physical systems, require far more sophisticated data processing in order to enable real-time diagnosis and control. 

This leads to the question of how to perform arbitrary (and likely complex) computational tasks using a distributed network of wireless sensors, 
each with limited resources both in energy and in processing capability.   
%
The answer seems to lie in two challenges. The first is the decomposition of  
complex computational tasks into smaller operations, each with its own input and output and collectively related through a certain dependency structure. 
The second is to distribute these operations among individual sensor nodes so as to incur minimal energy consumption and delay.  
This networked computing concept is a natural progression from the networked sensing paradigm.

Previous results on establishing the communication structure for in-network computation either consider only simple
functions like max/min/average/median/boolean symmetric functions~\cite{tag,median:podc,max:pods,wang:convergcast,kumar:boolean} that do not fully 
represent the complex computational requirements demanded by many practical engineering applications, or study asymptotic bounds 
which do not yield algorithms to determine the optimal communication structure for a given arbitrary network~\cite{kumar:scaling,kumar:function,medard:computing,srikant:function}.

In this paper we address the second challenge of finding the optimal communication structure for a given decomposition of a computational task.  While there are certain underlying commonalities, this is in general an application specific exercise dependent on the actual computation.  To keep our discussion concrete, we will study this problem within the context of structural health monitoring (SHM).  SHM is an area of rapidly growing interest due to the increasing need to provide low-cost and timely monitoring and inspection of deteriorating national infrastructure; it is also an appealing application of wireless sensor technologies.

The most common approach in SHM to detect damage is to collect vibration data using a set of wireless sensors in response to white/free input to the structure,
and then use the procedure of singular value decomposition (SVD) to determine the set of modes~\cite{mode1,mode2,mode3}. A {\em mode} is a combination of a frequency and a shape (in the form of a vector), 
which is the expected curvature (or displacement) of a surface vibrating at a mode frequency.   
In this study, we will use SVD as a primary example to illustrate an approach to determine how to perform such a complex computational task over a network of resource-constrained sensors.  
Compared to collecting raw vibration data (or the FFT of raw data) and performing the SVD computation at a central location, 
directly computing SVD within the network can result in significant reduction in both energy consumption and computational delay.
Conceptually, this reduction occurs because the output of SVD, a set of eigenvectors, is much smaller in size than its input, FFTs from individual sensor data streams, and evaluating multiple smaller SVDs in parallel is much faster than evaluating the SVD on the input from all sensors. 

In a recent result, Zimmerman \etal~\cite{andy:svd} proposed a method to decompose the SVD computation, a classical centralized procedure.
Here, we examine how to obtain an optimal communication structure corresponding to this decomposition, 
where optimality is defined with respect to minimizing energy consumption subject to a computational delay constraint. 
We show that this reduces to a generalized {\em clustering problem}, and here lies the generality of the results presented in this paper.  In essence, a centralized operation is decomposed into a number of computational elements (or operators) each operating on a set of inputs.  The optimization problem is then to figure out what set of inputs to group together (forming a cluster), and what relationship needs to be maintained (in the form of data sharing or message passing) between clusters according to the specification of the decomposition.  
For this reason, while SVD is the example throughout our discussion, the methodology used here is more generally applicable.  Note also that SVD itself is an essential ingredient in a broader class of signal processing algorithms, including classification, identification and detection~\cite{andy:svd,svd1,svd2,svd3,svd4,svd5}.


Our main results are summarized as follows.  
\begin{enumerate}
\item We formally define the above networked computing problem and establish that it is NP-hard (in Section~\ref{sec:prelim}). 
\item We derive a lower bound by relaxing the delay constraint and show that the optimal solution to the unconstrained problem has a simple structure that sheds light on the original problem (in Section~\ref{ref:lower}). 
\item We present an integer linear program (ILP) to solve the constrained problem as well as an approximate algorithm with a proven 
approximation ratio (Section \ref{ref:multiple}); we also present a distributed version of the approximate algorithm (Section \ref{ref:multiple}).  
\item We use both simulations and experiments to evaluate our algorithms (Section \ref{sec:simulations}). 
\end{enumerate}


We end this introduction with a simple example to illustrate that different computational objectives will have different optimal communication structures.
We compare the optimal routing for data compression, and for computing SVD. 
Assume that compression converts $2$ input streams of size $R$ bits each to an output stream of $R + r$ where $r<R$ ~\cite{bazarul:ton}.
The SVD operator, as discussed in detail in Section~\ref{ref:multiple}, converts $k$ input streams of size $R$ bits each, 
into $k$ eigenvectors of size $r$ bits each with $r < R$. 
Consider the simple 4-node topology of Figure~\ref{fig:comp_svd} and 
the two possible communication structures, with node $0$ being the base station and assume all links are of unit length/cost. 
As derived in~\cite{bazarul:ton}, data compression requires an exchange of $3R + 3r$ (using successive encoding) and $4R + r$ bits respectively for the communication structures (a) and (b). Hence, if $R > 2r$, then (a) is better. 
On the other hand, in the case of SVD if we do not perform in-network computation, then sending all raw data to node $0$ results in a cost of $6R$ and $5R$ over the two structures, respectively. If we perform in-network computation, then as detailed in Section \ref{ref:multiple}, the resulting costs are $3R+6r$ and $3R+3r$ for the two structures, respectively.  
Hence (b) is always better for the SVD computation.   
\begin{figure}[htb]
 \centering
\includegraphics[width=2.6in]{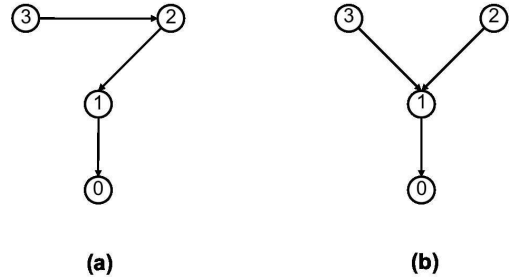}
\vspace{-0.1in}
  \caption{In-network computation and compressed sensing can have a different optimal communication structure. (a) and (b) represent the two
possible communication structures for a simple 4-node topology.} 
\vspace{-0.3in}
\label{fig:comp_svd}
\end{figure}

\section{Problem Formulation and Preliminaries}
\label{sec:prelim}
In this section, we first introduce the relevant background on structural health monitoring, then present the network model
and formally introduce the problem.

\subsection{Background on Structural Health Monitoring}
\label{sec:shm}
During the past two decades, the SHM community has become increasingly focused on the use of the structural vibration data
to identify degradation or damage within structural systems. The first step in determining if the vibration data collected by a set of sensors 
represents a healthy or an unhealthy structure is to decompose the spectral density matrix into a set of single-degree-of-freedom systems. 
Assuming a broadband white input to the system, this can be accomplished by first obtaining an estimate of the output power spectral density (PSD) matrix 
for each discrete frequency by creating an array of frequency response functions using the Fast Fourier Transform (FFT) information from each
degree of freedom. Early studies in this field focused on identifying changes in modal frequencies or the eigenvalues of the PSD 
matrix using the peak picking method~\cite{pp2} to detect damage in large structural systems~\cite{pp1}. More recent studies have observed that viewing changes in modal frequencies 
in combination with changes in mode shape information (eigenvector of the PSD matrix) makes it increasingly possible to both detect and locate
damage within a variety of structural types and configurations~\cite{mode1,mode2,mode3}. One of the most widely used method for mode shape estimation is
the frequency domain decomposition (FDD) method proposed by Brincker \etal~\cite{fdd}. This method involves computing the SVD\footnote{Please 
see Appendix~\ref{appendix2} for a description of the SVD computation.} of the PSD matrix to extract the eigenvectors/mode shapes. 

The most common implementation of the FDD method over a wireless sensor network is to have each sensor send its vibration data to a central sensor node
which computes the SVD of the PSD matrix and then distributes the mode shapes back to each sensor. This method requires significant computational
power and memory at the central node as well as a significant energy consumption in the network to communicate all this data to the central node. 
For example, if there are $100$ sensor nodes in the network, this implementation requires the central sensor node to compute the SVD of a $100 \times 100$ PSD matrix
as well as having each of the $100$ sensor nodes send all their vibration data to one central node. 

Zimmerman \etal~\cite{andy:svd} proposed an alternative
implementation by decomposing the computation of SVD 
(graphically represented in Figure~\ref{fig:svd_computation}). 
Each sensor node is assumed to be aware of the eigenvalues of the PSD matrix (which have been 
determined using the peak-picking method
\footnote{Finding the optimal communication structure to implement peak-picking in a distributed manner
turns out to be the same as in the case of data compression~\cite{bazarul:ton,sivakumar:secon,pattem:ipsn}, and is thus not studied here.}) 
and the FFT of its own sensed data stream. 
%
Denoting the entire set of nodes as $V$, if a sensor has the FFT of $N \subset  V, | N | > 1$ sensors and all the eigenvalues, then it can compute the SVD of the PSD matrix
using $ | N |$ sets of FFT results and determine $ | N |$ eigenvectors. Let another sensor node be in possession of the FFT of 
$N' \subset V, | N' | > 1$ sensors. It can perform a similar computation to determine $ | N' |$ eigenvectors.
To be able to combine results from these two computations to construct the $ | N \cup N' | $ eigenvectors,
one needs to be able to determine the appropriate scaling factors.  This notion is precisely given in the following. 

\begin{definition}
Two computations are called \emph{combinable}
if one can determine the appropriate scaling factors to combine them.
A computation on $N$ nodes and another computation on $N'$ nodes is combinable if and only if either $N \cap N' \neq \phi$
(that is, there is at least one common sensor in $N$ and $N'$), or there exists another
computation on $N''$ nodes which is combinable with both $N$ and $N'$.   
\end{definition}

\begin{figure}[htb]
 \centering
\includegraphics[width=2.9in]{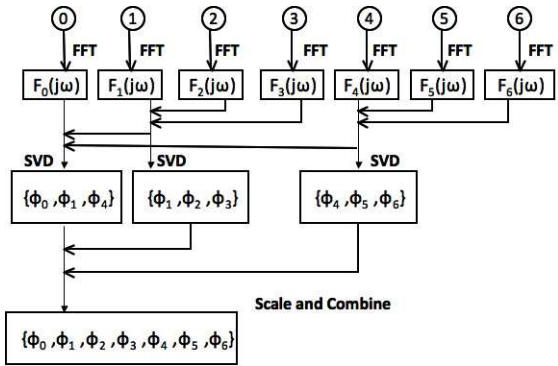}
  \caption{Decomposing the computation of SVD using in-network computation.}
\vspace{-0.15in}
\label{fig:svd_computation}
\end{figure}

If $R$ denotes the size in bits required to represent the FFT of a sensor stream and $r$ denotes the
size in bits to represent a eigenvector, each SVD computation which combines the FFT of $k$ sensor streams
reduces the number of bits from $kR$ to $kr$. Note that the size of the output stream does not
depend on $R$ but only on $r$, which depends on the size of the network.


\subsection{A Generalized Clustering Problem}
With the above decomposition, it can be seen that the associated communication problem can be cast as a {\em generalized clustering problem}:  the solution lies in determining which subset of sensors ({\em cluster}) should send their FFTs to which common node ({\em cluster head}), who then computes the SVD for this subset, such that these subsets have the proper overlap to allow individual SVDs to be scaled and combined.  We consider this clustering problem generalized because due to the combinability requirement clusters will need to overlap, e.g., the cluster head of one cluster can also be the member of another cluster.  The resulting hierarchy is thus driven by this requirement rather than pre-specified as in many other clustering approaches (e.g., the two-layer clustering in LEACH \cite{leach}).  However, should this requirement be removed then the clustering problem becomes a fairly standard one. 

\subsection{Network Model and Problem Definition}
\label{sec:model}

We will proceed to assume a network model of an undirected graph denoted as $G(V, E)$. 
Each (sensor) node in $V$ acts as both a sensor and a relay.
If two nodes can successfully exchange messages directly with each other, there exists an edge $e \in E$ between them.
Let there be a weight $w_e \geq 0$ associated with each edge which denotes the energy expended in sending a packet across this edge. 
Without loss of generality, we take node $0$ to be the central node or the base station. 
We also assume that all sensors (including the base station) are identical in their radio capability (and hence have the same energy consumption per bit). 
This is done to keep the presentation simple and can be easily relaxed. 

Each node has a local input vibration stream. The goal is to evaluate the SVD of the PSD matrix formed by the input vibration streams of all the sensors.  We define a {\em sensing cycle} to be the time duration in which each sensor performs the sensing task to generate a vibration stream, the SVD is then computed and the mode shapes are made known at the base station.  The length of this cycle as well as how this procedure may be used in practice are discussed in Section \ref{sec:disc_aggr}.  
Our objective 
is to determine the optimal communication structure to minimize the energy consumption in a sensing cycle under a constraint on the maximum duration of a sensing cycle.



The two metrics of interest, namely energy consumption and computational delay are precisely defined as follows. 

{\em Energy consumption} is defined as the total communication energy consumed in the network in one
sensing cycle. Let $E_{Tx}$ and $E_{Rx}$ denote the energy consumed in transmitting and receiving a bit of data.  For convenience we will 
denote 
$E_b= E_{Tx} + E_{Rx}$.    
Then the energy consumed in transmitting a packet of $B$ bits over an edge $e$ is $w_e B E_b$ \footnote{ 
Our algorithms and analysis 
do not depend on the exact model used for energy consumption, provided that it remains a function of the number of bits transmitted;  
more is discussed in Section~\ref{sec:disc_aggr}.}. 

{\em Computational delay} is defined as the time it takes to compute a designated function at a sensor node. 
As observed in~\cite{nsdi:wishbone,andy:svd}, the computational time is the chief contributor to delay as packet 
sizes in sensor systems tend to be very small. Thus, the duration of a sensing cycle depends primarily on the maximum computational
delay amongst all sensor nodes. 

A constraint on the computational delay essentially translates into a constraint on the maximum number of FFT's which can be combined at a node. 
%
We now formally introduce the problem. 

{\bf Problem P1}: Find (1) the set $S$ of sensor nodes on which the SVD computation will take place, 
(2) for each $s\in S$, their corresponding set $N_s$ of sensor nodes whose FFT will be made available at $s$, and
(3) a routing structure, so as to minimize $E$, the total energy consumed,  
subject to the constraint that $|N_s | \leq n_s$, $\forall s \in S$, where $n_s$ denotes the maximum cluster size allowed at node $s$ that 
corresponds to its computational delay constraint, and that the computations on all pairs $s_1, s_2 \in S$ are combinable. 

The set $S$ will also be referred to as the set of cluster heads, and set $N_s$ the cluster associated with head node $s$. 
In the above description we have imposed individual delay constraints.  Note that the computational delay of one round of SVD is dominated by the 
largest delay among all nodes if the computations of successive rounds are pipelined. 
One could also try to minimize the maximum computational delay with 
a constraint on the energy consumption. Indeed, it can be shown that the dual of the linear programs we propose will optimally solve this alternative formulation. 

A decision version of P1 can be shown to be NP-hard through a reduction from set cover. The proof is given in Appendix \ref{appendix1}. 
\begin{theorem}
\label{thm:np}
There is no polynomial time algorithm that solves P1, unless $P=NP$.
\end{theorem}

\section{A Lower Bound on the Value of P1}
\label{ref:lower}

%
To simplify presentation, in this section we will assume that the weights of all edges are equal. This is not restrictive
as all bounds derived in this section can be easily modified
to incorporate different weights. With this assumption, the energy consumed in sending data
from one node to another merely depends on the number of hops on the path between them.

\begin{definition}
A {\em data collection tree (DCT)} for $G(V,E)$ is the spanning tree such that the path from each node $v \in V$
to the base station has the minimum weight. 
\label{def:dct}
\end{definition} 

Compared to a minimum spanning tree (MST), a DCT offers minimum weight on each path to the root rather than 
over the entire tree.  Since all weights are equal, a path of minimum weight is equivalent
to that of minimum hop count. 
Let $d_0(v)$ denote the hop count of node $v \in V$ in the DCT. 

The following lemma provides a lower bound on the minimum energy consumption for P1 given any choice of $S$.
\begin{lemma}
\label{lemma:lb_no_comp}
Consider P1 defined on graph $G(V, E)$, and a set of cluster heads $S \neq \phi$,
then a lower bound on the optimal energy consumption, denoted by $E(S)$, is given by 
\begin{eqnarray}
E(S) \geq \left( \left( | V | - 1 \right) R + \sum_{v \in V } \left( d_0(v) - 1 \right) r 
+ |S|  r \right) E_b . 
\end{eqnarray}
\end{lemma}

\begin{proof}
For all nodes $v \in V \backslash S$, a message of size $R$ (containing the FFT) needs to be transmitted from $v$ to some node in $S$. 
This message goes over at least one hop to reach this node, after which the size reduces to $r$ (the eigenvector). 
Since the minimum hop count from $v$ to the base station is $d_0(v)$, if the message of size $R$ goes over one hop, the message of size $r$ will go over at least $d_0(v) - 1$ hops.  As $R>r$, the amount of transmission required includes $|V| - |S|$ transmissions of size $R$ and $\sum_{v \in V \backslash S} (d_{0}(v) - 1)$ transmissions of size $r$. 

In addition to the above, each of the $|S|$ computations needs to be combinable.  This means that $\forall s_1, s_2 \in S$, either $N_{s_1} \cap N_{s_2} \neq \phi$ 
or there exists another node $s_3 \in S$ such that the computations at nodes $s_1$ and $s_3$, as well as that at nodes $s_2$ and $s_3$ 
are combinable, respectively. 
To understand how many extra messages are needed to satisfy this constraint, we construct the following
graph $G^S(S, E^S)$: an edge is added to $E^S$ between nodes $s_1$ and $s_2$, $s_1, s_2 \in S$, only if $N_{s_1} \cap N_{s_2} \neq \phi$. 
Each edge in this graph thus represents at least one common node between $N_{s_1}$ and $N_{s_2}$; each common node needs to transmit a message of size $R$ to both $s_1$ and $s_2$.  

It follows that each edge implies at least one extra transmission of size $R$ in addition to the 
$|V| - |S|$ transmissions of size $R$ computed earlier. 
Two nodes in $s_1, s_2 \in S$ are combinable if and only if there exists a path
between $s_1$ and $s_2$ in $G^S(S, E^S)$. For a path to exist between every pair of nodes, 
$G^S(S, E^S)$ needs to have at least $|S|-1$ edges. This means at least $|S|-1$ extra transmissions of 
size $R$ are required for all pairs $s_1, s_2 \in S$ to be combinable. 
Taking this into account, at least $|V|-1$ transmissions of size $R$ and $\sum_{v \in V \backslash S} (d_{0}(v) - 1)$ transmissions of size $r$ have to take place.  

Finally, computed eigenvectors from nodes $v \in S$ each goes through at least $d_0(v)$ hops. 
Combining all of the above yields
\begin{eqnarray}
E(S) &\geq& \left( \left( | V | - 1 \right) R + \sum_{v \in V \backslash S} \left( d_0(v) - 1 \right) r 
+ \sum_{v \in S} d_0(v) r \right) E_b \nonumber \\
&=& \left( \left( | V | - 1 \right) R + \sum_{v \in V} \left( d_0(v) - 1 \right) r  + |S| r \right) E_b ~. \label{eqn:lower_bound}
\end{eqnarray}
\end{proof}

An interesting observation is that the lower bound only depends on the size of $S$ and not its membership. 
One way to get close to this bound is to limit the delivery of any FFT to a single hop and
route the FFT and the subsequent eigenvector along shortest paths.  This motivates a particular solution for any given tree structure. 

\begin{definition}
Consider a graph $G(V, E)$ and a routing tree $T$.  Define a communication structure $A_{P2}(T)$ as follows: 
(1) All non-leaf nodes in $T$ constitute the set $S$, (2) cluster $N_s, s \in S$ consists
of all immediate children of $s\in S$, and (3) each node sends its own FFT to its parent node on $T$, and a node $s \in S$
sends eigenvectors for itself as well as its children along $T$ to the base station.
This will be referred to as {\em tree solution $T$}. 
\end{definition}

We next consider an unconstrained version of P1, i.e., by removing the computational delay constraint. We refer to this unconstrained problem as {\bf P2}.
\begin{lemma}
\label{lemma:tree_energy}
Consider P2 defined on a graph $G(V, E)$, and a routing tree denoted by $T$ defined on the same graph.
Let $d_T(v)$ denote the hop count of node $v \in V$ in $T$.  
Then tree solution $A_{P2}(T)$ is feasible and has an energy consumption  
\begin{equation}
E_{A_{P2}}(T) = \left( \left( | V | - 1 \right) R + \sum_{v \in V} \left( d_T(v) - 1 \right) r  + |S| r \right) E_b ~.
\end{equation}
\end{lemma}

\begin{proof}
We first show feasibility, i.e., each pair of nodes $s_1, s_2 \in S$ are combinable. 
Since $S$ consists of all non-leaf nodes on a tree, there exists a path between any pair of these nodes.  Thus $A_{P2}(T)$ is feasible. 

Next since each node (except for the base station) sends its FFT to its parent, this results in a cost of $\left( |V|-1 \right) R E_b$; 
each non-leaf node computes the SVD from its children's FFT and its own, and then sends the eigenvectors to the base station, resulting in  
$\sum_{v \in V} \left( d_T(v) - 1 \right) r$  + $| S | r$ bits.  
Putting everything together yields the lemma.
\end{proof}

This lemma suggests that of all solutions given by a tree structure, the one that minimizes both $d_T(v)$ and $|S|$ will result in the smallest
energy consumption.  
This motivates the construction of a DCT (which minimizes $d_T(v)$) that has a minimum number of non-leaf nodes (which minimizes $|S|$). 

\begin{definition}
A {\em minimum non-leaf node data collection tree}, or MDCT, defined on graph $G(V, E)$ is a DCT that has the smallest number of non-leaf nodes among all DCTs defined on $G(V, E)$.  We will denote this tree as $T_M$. 
\end{definition}

A key property of an MDCT is that it is impossible to move all the children of wa non-leaf node $v \in V$ on $T_M$ to other {\em non-leaf nodes} 
of height $ \leq d_T(v)$.  This is because if this could be done then we can effectively reduce the number of non-leaf nodes on $T_M$, which is a contradiction. 
Figure \ref{fig:non_leaf_min} gives an example: both (a) and (b) are DCTs on the same graph, but the former is not a MDCT while the latter is. 
\begin{figure}[htb]
 \centering
\includegraphics[width=2.9in]{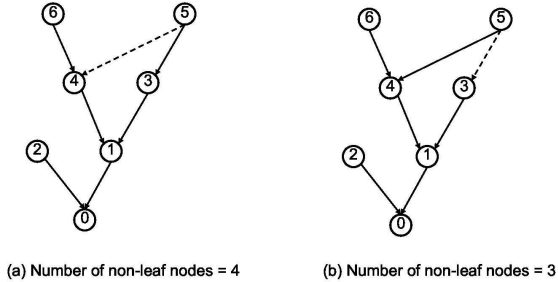}
  \caption{Two data collection trees for the same network. The solid lines represent the edges of the tree. $T_M$ is the tree in (b).}
\label{fig:non_leaf_min}
\vspace{-0.15in}
\end{figure}




\begin{theorem}
\label{thm:without_delay}
Consider P2 defined on $G(V, E)$, and an associated MDCT $T_M$ with cluster head set $S$. Under the condition $R>2r$,
an optimal solution to P2 is given by $A_{P2}(T_M)$.\footnote{The condition $R>2r$ is easily satisfied in SVD computation for SHM.}
\end{theorem}

\begin{proof}
%
%
%
From Lemma \ref{lemma:tree_energy} we know that $E_{A_{P2}}(T_M)$ matches exactly the lower bound given in (\ref{eqn:lower_bound}). 
Consider any other solution with a cluster head set $S'$ such that $|S'|\geq|S|$.  By Lemma \ref{lemma:lb_no_comp} 
$E(S')\geq E_{A_{P2}}(T)$ so any solution with a larger set $S'$ is no better. 

Consider next any solution with a set $S''$ such that $|S''| < |S|$.  
By Lemma \ref{lemma:lb_no_comp}, using $S''$ instead of $S$ reduces the energy by no more than $\left( |S|-|S''| \right) r E_b$.
On the other hand, consider a node $v\in S$ and $v\not\in S''$.  By the property of the MDCT $T_M$, there are only three possibilities in how the children of $v$ can send their FFT under the new solution $S''$: 
(1) each child of $v$ sends its FFT to some node $v''\in S''$ with $d_0(v'') = d_0(v)$ via a single hop; (2) at least one child of $v$ sends its FFT to a $v''\in S''$ with $d_0(v'') > d_0(v)$ via a single hop; (3) at least one child of $v$ sends its FFT over at least $d\geq 2$ hops to reach $v''\in S''$ with $d_0(v'')+d \geq d_0(v)+1$. 
Denote these sets as $V_1$, $V_2$ and $V_3$, respectively.  Note that $|S| = |S \cap S''| + |V_1| + |V_2\cup V_3|$. 

In case (1), at least one such $v''\in S''$ cannot be in $S$, and these $v''$ nodes will be distinct for different $v \in S, \not\in S''$ nodes, 
for otherwise it contradicts the definition of an MDCT.  Thus for each such $v\in S, \not\in S''$ there corresponds a $v''\in S'', \not\in S$.  
Therefore case (1) does not contribute to any reduction in energy consumption compare solution $S''$ to $S$. Thus, $|S''|<|S|$ can only be true 
if either (2) or (3) is true for some $v\in S, \not\in S''$; in other words, $|S|-|S''|=|V_2\cup V_3|$.   For each such $v$, if it falls under 
case (2) then there is an energy increase (from $S$ to $S''$) of at least $r E_b$ due to the height increase of $v''$ over $v$; if it falls under 
case (3), the energy increase is at least $(R-r)E_b>rE_b$ by the condition stated in the theorem.  Thus the total energy increase is at least $r E_b$ 
for each $v\in V_2\cup V_3$; therefore the increase is at least
%
$\left( |S|-|S''| \right) r$.  
Hence any solution with a smaller set $S''$ is no better, completing the proof. 
\end{proof} 

To summarize, an MDCT yields the optimal solution
for P2, which also serves as a lower bound to the value of P1.  Note that in this solution the overlap between clusters is through cluster heads;  
all cluster heads (except for the base station) is a member of another cluster.  

\section{Exact and Approximate Algorithms}
\label{ref:multiple}

We next present an integer linear program (ILP) to solve P1 exactly and a $O \left( log \left( |V| \right) \right)$ approximation algorithm for P1.

\subsection{An Exact ILP for P1}
\label{ilp:original}
We first introduce optimization variables used in the ILP. 

The set of variables $x_{ij}, i, j\in V$ define both the sets $S$ and $N_s, \forall s\in S$ as follows. 
$x_{ij} := 1$ if the FFT of node $i$ is evaluated at node $j$ (i.e. $i \in N_j$), and $0$ otherwise.
$x_{ii} := 1$ if $i \in S$ and $0$ otherwise. 

$p_{ijk}:= 1$ if the FFT of node $k$ is evaluated at both nodes $i$ and $j$.  This notation is used for convenience of presentation
only as it is completely determined by $x_{ij}, i, j\in V$. 


Finally, the variables $c_{ijn}$ recursively verifies the combinability relationship between two nodes $i, j\in S$ as follows:   
\begin{equation}
\label{eqn:cijn}
c_{ijn} = \left\{ \begin{array}{c l} 1 & \mbox{if } n=0 \mbox{ and } \sum_{k \in V} p_{ijk} \geq 1,  \\
1 & \mbox{if } 0<n<|V| \mbox{ and } \\
& \sum_{k \in V} c_{ik(n-1)}.c_{jk(n-1)} + c_{ij(n-1)} \geq 1, \\
0 & \mbox{otherwise.}
\end{array} \right.
\end{equation} 
Thus $c_{ij0}=1$ if the pair $i, j \in S$ share common nodes in their respective clusters, $c_{ij1}=1$ if the pair $i, j$ either share common nodes directly or each shares common nodes with a common third cluster, and so on.  If the pair $i, j \in S$ are combinable, we will have $c_{ij(|V|-1)} = 1$. 

Finally, 
$W_{ij}$ denotes the sum weight of all edges along the shortest path from node $i$ to node $j$.

The ILP below solves P1 exactly, where the minimization is over the choice of $x_{ij}, \forall i, j\in V$. 
\begin{eqnarray}
\label{eqn:obj} & \mbox{\bf (ILP\_P1)} ~~~ \mbox{min } \sum_{i \in V, j \in V} x_{ij} E_b \left( R W_{ij} 
+ r W_{j0} \right) & \\ 
& \mbox{s.t.} & \nonumber \\
\label{eqn:c1} &  \sum_{j \in V, j \neq i} \frac{x_{ji}}{V} \leq x_{ii} \leq \sum_{j \in V, j \neq i} x_{ji}, \forall i \in V \\
\label{eqn:c2} & \sum_{j \in V} x_{ij} \geq 1, \forall i \in V \\
\label{eqn:c3} & p_{ijk} \leq \frac{x_{ki} + x_{kj}}{2},  \forall i,j,k \in V \\
\label{eqn:c4} & c_{ij0} \leq \sum_{k \in V} p_{ijk}, \forall i,j \in V& \\
\label{eqn:c5} & c_{ij(|V|-1))} \geq x_{ii} + x_{jj} - 1,   \forall i,j \in V & \\
\label{eqn:c6} & t_{ijkn} \leq \frac{c_{ik(n-1)} + c_{jk(n-1)}}{2}, \forall i,j,k \in V, 0 < n < |V|  & \\
\label{eqn:c7} & c_{ijn} \leq c_{ij(n-1)} + \sum_{k \in V} t_{ijk(n-1)}, \forall i,j \in V, 0 < n < |V| & \\
\label{eqn:c8} & c_{iin} = 0, \forall i \in V, 0 \leq n < | V |    & \\
\label{eqn:c9} & \sum_{i \in V} x_{ij} \leq n_j, \forall j \in V   & \\
\label{eqn:c10} & x_{ij}, c_{ijk}, p_{ijk}, t_{ijkn} \in \{0,1\} \forall i,j,k \in V, 0 \leq n < | V |  & 
\end{eqnarray}
The objective (Eqn (\ref{eqn:obj})) is fairly straightforward: if the FFT of node $i$ is sent to 
node $j$, it costs $R W_{ij} E_b$.  
The FFT from $i$ produces a unique eigenvector of size $r$ at node $j$ as a result of this SVD computation, which costs  
$r W_{j0} E_b$ to send to the base station.

The first constraint (Eqn (\ref{eqn:c1})) sets the value of $x_{ii}$ to $1$ if $N_i \neq \phi$, and $0$ otherwise 
(note that if $N_i \neq \phi$, then $1 \leq \sum_{j \in V, j \neq i} x_{ji} \leq |V|$). 
Eqn (\ref{eqn:c2}) ensures that the FFT of every sensor node is sent to at least one node.
Eqn (\ref{eqn:c3}) ensures that $p_{ijk}=1$ if the FFT from node $k$ is sent to
both nodes $i$ and $j$. 

The next five constraints ensure the combinability of the solution by limiting the value of $c_{ijn}$. 
Eqn (\ref{eqn:c4}) ensures that $c_{ij0}=1$ if there is at least one node common to $N_i$ and $N_j$. 
Eqn (\ref{eqn:c5}) states that if both $i,j \in S$, the computations at $i$ and $j$ should be combinable.  
Eqns (\ref{eqn:c6}) and (\ref{eqn:c7}) populate the value of $c_{ijn}$. Note that $t_{ijkn}$ is a temporary variable introduced to
express the quadratic condition in Equation (\ref{eqn:cijn}) as a linear function. 
Note that the presence of Eqn (\ref{eqn:c5}) forces Eqns (\ref{eqn:c3}), (\ref{eqn:c4}), (\ref{eqn:c6}), and (\ref{eqn:c7}) to assign the maximum possible value to the LHS; 
similarly, the presence of the latter forces (\ref{eqn:c5}) to assign the minimum possible value to the LHS. 

Eqn (\ref{eqn:c8}) sets the value of $c_{iin}$ to zero for every $i \in V, 0 \leq n < | V |$.
This prohibits a corner case where $c_{ijn}$ is set to $1$ by setting $c_{ii(n-1)}$ to $1$ without ensuring
that the computation at $i$ and $j$ are combinable.  
Finally, Eqn (\ref{eqn:c9}) imposes the computational delay constraint at each sensor node. 

\subsection{Degree-Constrained DCT: Problem P3}

In this and the next two subsections we will develop a $O \left( log \left( |V| \right) \right)$ approximation to the optimal solution of P1. 
To simplify the presentation, we will again assume that all edge weights are equal. Note that all the algorithms proposed in this section 
can be easily modified without changing their approximation factors to incorporate different weights. 


The basic idea is to first use a DCT to find a feasible solution to P1. 
A feasible solution requires that each cluster is size-limited due to the computational delay constraint: a node $v$ cannot have more than $n_v-1$ immediate children.  This leads to the following definition. 
\begin{definition}
A {\em degree-constrained data collection tree}, or DDCT, is a tree $T$ which minimizes $\sum_{v \in V} d_T(v)$ under the constraint that a 
node $v \in V$ has no more than $n_v-1$ immediate children, where $n_v, \forall v \in V$ are given constants.
\end{definition}


{\bf Problem P3}: Find a DDCT for $G(V,E)$, which in turn determines the set $S$, clusters $N_s, \forall s\in S$, and the routing structure. 

That a solution to P3 is feasible for P1 is obvious, but it may not be optimal for P1, even if it has the fewest non-leaf nodes among all DDCTs because a node may no longer be on its shortest path. 

It's worth noting that P3 is also NP-hard; it is APX-hard even when weights on edges satisfy the triangle inequality~\cite{tree:apx}.
Results on P3 are known only for complete graphs whose weights satisfy the triangle inequality~\cite{tree:apx,helmick:multicast}. 
To the best of our knowledge our work here is the first to propose algorithms with proven approximation ratios for 
P3 in graphs induced by a communication network.

We proceed as follows.  We first present an ILP (ILP\_P3) to solve P3 exactly. 
This ILP has much fewer variables and constraints than ILP\_P1, 
and hence takes less time to solve.  We then relax ILP\_P3 to an LP and solve it via appropriate rounding of fractional values. 
This rounding algorithm, referred to as algorithm {\bf LPR}, is thus an approximation algorithm for P3, and therefore also an approximation algorithm for P1. 
We derive the approximation factor for LPR with respect to problem P1 in Theorem~\ref{thm:approx_aggr}.
Finally, based on the intuition derived while analyzing LPR, we present a simpler, distributed
approximation algorithm with the same asymptotic approximation factor.

\subsection{An ILP for Problem P3}
\label{sec:dtree-ilp}

We define the following variables used in the ILP in finding a DDCT.  For a given $G(V,E)$, define a graph $\bar{G}(V,\bar{E})$  
with directed edges, by replacing each undirected edge in $E$ with two directed edges, one in each direction. 
Let $O_v, v \in V$ denote the set of outgoing edges from node $v$ in $\bar{E}$. 
Similarly, let $I_v, v \in V$ denote the set of incoming edges into node $v$ in $\bar{E}$.

The set of variables $x_e, e \in \bar{E}$ define whether an edge is on the DCT as follows. 
$x_e := 1$ if edge $e$ is on the DCT, and $0$ otherwise. 
The variable $f_e, e \in \bar{E}$ will be referred to as the {\em flow value} over the edge $e$; it denotes the 
number of nodes using edge $e$ to reach the base station on the DCT.  $f_e=0$ if edge $e$ is not on the DCT. 

The following ILP solves P3 exactly, where the minimization is over the choice of $x_e, \forall e \in \bar{E}$. 
\begin{eqnarray}
\label{eqn:d_obj} & \mbox{\bf (ILP\_P3)} ~~~ \mbox{min } \sum_{e \in \bar{E}} f_e & \\
\label{eqn:d_c1} & \sum_{e \in {I_0}} f_e - \sum_{e \in O_0} f_e = | V | - 1 & \\
\label{eqn:d_c2} & \sum_{e \in {I_v}} f_e - \sum_{e \in O_v} f_e = -1, \forall v \in V \backslash \{0\} & \\
\label{eqn:d_c3} & f_e \leq \left( | V | - 1 \right) x_e , \forall e \in \bar{E} & \\
\label{eqn:d_c4} & \sum_{e \in \bar{E}} x_e = | V | - 1 & \\
\label{eqn:d_c5} & \sum_{e \in O_v} x_e = 1, \forall v \in V \backslash \{0\} & \\
\label{eqn:d_c6} & \sum_{e \in I_v} x_e \leq n_v-1, \forall v \in V  & \\
\label{eqn:d_c7} & x_e \in \{0, 1\}, \forall e \in \bar{E} & \\
\label{eqn:d_c8} & f_e \in \{0, 1, \ldots, | V | - 1\}, \forall e \in \bar{E}  & 
\end{eqnarray}
The objective function minimizes the total flow, which essentially minimizes $\sum_{v \in V} d_T(v)$. 

The first two constraints ensure that each node sends a unit flow towards the base station.
The third constraint forces $f_e$ to be $0$ if $x_e$ is $0$, otherwise, it is redundant. 
%
Eqn (\ref{eqn:d_c4}) ensures that the output has exactly $|V|-1$ edges. 
Eqns (\ref{eqn:d_c5}) and (\ref{eqn:d_c6})) ensure that
there is no more than one outgoing edge per vertex (other than the base station) and no more than $n_v-1$
incoming edges into vertex $v$.  Eqns (\ref{eqn:d_c4}) and (\ref{eqn:d_c5}) together ensure that the output is a tree and Eqn (\ref{eqn:d_c6}) 
ensures that a node $v$ has no more than $n_v-1$ immediate children. 

\subsection{Algorithm LPR: an LP Rounding Approximation}

We next present a polynomial-time approximation algorithm which relaxes ILP\_P3 
to a linear program (LP), by allowing $0\leq x_e\leq 1$ and $f_e\geq 0$ to be fractional and appropriately rounding the fractional values.  
This algorithm is referred to as LPR and shown in Figure \ref{algo:with_aggr}. 

\begin{figure}[ht]
  \begin{center}
  \begin{ttfamily}
    \begin{footnotesize}
      \flushleft
       $NV = \{0\}$, $NE = \phi$, $h=0$, assign $h_v = -1,$ $\forall v \in V\backslash \{0\}$ and $h_0=0$\\
       while $(NV != V)$ do \\
       \hspace*{0.2in} $h=h+1$ \\ 
       \hspace*{0.2in} Solve the ILP for P3 with fractional variables \\ 
       \hspace*{0.2in} and the additional constraint that $x_e = 1, \forall e \in NE$ \\
       \hspace*{0.2in} For $\forall v \in NV$ and $h_v = h-1$ \\
       \hspace*{0.4in} If the value of $x_e$ for more than $(n_v-1)$ \\
       \hspace*{0.4in} incoming edges at $v$ is greater than $0$ \\ 
       \hspace*{0.6in} Set the largest $(n_v-1)$ $x_e$ values \\ 
       \hspace*{0.6in} amongst the incoming edges at $v$ to $1$ \\
       \hspace*{0.6in} (ties are broken arbitrarily) \\
       \hspace*{0.4in} Otherwise \\
       \hspace*{0.6in} Set the $x_e$ value of all incoming \\
       \hspace*{0.6in} edges at $v$ to $1$ \\
       \hspace*{0.4in} Add the edges for which $x_e$ was set to $1$ \\
       \hspace*{0.4in} in the previous step to $NE$ \\
       \hspace*{0.4in} For all edges added to $NE$ in the \\
       \hspace*{0.4in} previous step, add the node $v$ from \\
       \hspace*{0.4in} which the edge emanates to $NV$ and \\
       \hspace*{0.4in} assign $h_v = h$ \\
\end{footnotesize}
\end{ttfamily}
\end{center}
\vspace{-0.15in}
  \caption{Algorithm LPR: The LP rounding approximation algorithm for P3.}
\vspace{-0.25in}
  \label{algo:with_aggr}
\end{figure}

\subsection{The Approximation Factor of LPR}
\label{sec:approx}

Even though LPR makes no assumptions on the network, our derivation
of the approximation factor assumes the following: (1) $n_{min} \geq 3$, where $n_{min} = \mbox{min}_{v \in V} n_v$, 
(2) the unconstrained MDCT constructed over $G(V, E)$ has a height of 
$O\left( \log (|V|)\right)$, 
and (3) nodes can transmit to each other if the distance between them is less than a transmission range $R_{tx}$. 

To understand (1), note that if $n_v = 2,  \forall v \in V$, then each node has at most 1 child and the constructed tree is thus linear (a chain).  The problem subsequently reduces to the traveling salesman problem.  Similarly, if most nodes disallow more than 1 child, the resulting tree will be close to linear, which is not a very interesting routing structure to study. Finally, and more importantly, most existing sensor platforms have sufficient computational power to quickly combine FFT's from at least 3 nodes, easily satisfying this assumption. 
%
Assumption (2) is also easily satisfied as sensor networks used for SHM are in general not very sparse. 
Assumption (3) is very commonly adopted for analytical tractability. 
However, our analysis is not heavily dependent on this assumption (more is discussed in the footnote in the proof 
of Lemma~\ref{lemma:hgt_log}) and the same approximation factor also holds under more realistic physical layer assumptions.

We next derive the approximation factor of LPR with respect to P1.
The analysis is based on the observation that the approximation factor
is essentially the ratio between the height of the DDCT constructed using LPR 
and the height of the MDCT (discussed in more detail in the proof of Theorem~\ref{thm:approx_aggr}).  


Denote the height of the MDCT 
by $h_{orig}$ 
and the height of the DDCT generated by algorithm LPR $h_{ddct}$. Define a {\em non-full node} to be a node $v$ at height
$h < h_{ddct}$ which has less than $n_v-1$ children. A height $1 \leq h < h_{ddct}$ is defined to be a {\em non-full height} 
if there exists at least one non-full node at height $h$.  We then have the following lemma. 
\begin{lemma}
\label{lemma:non-full}
Consider running algorithm LPR on a set of $m$ nodes with a randomly selected base station and a topology such that
the maximum set of nodes that cannot transmit to each other has a size $p$.  
Then the resulting DDCT cannot have more than $p$ non-full heights. 
\end{lemma}

\begin{proof}
We prove this by contradiction. Let there be $p+1$ non-full heights: $h_1 < \ldots < h_{p+1}$. Let $v_{i}$ be a
non-full node at height $h_{i}, 1 \leq i \leq p+1$. Then, $v_i, v_j$, $1 \leq i < j \leq h_{p+1}$ cannot
transmit to each other, for otherwise LPR would have labeled $v_j$ as the child of $v_i$. Thus none of the nodes
$v_1, \ldots, v_{p+1}$ can transmit to each other. However, by assumption we cannot have more than $p$ 
nodes which cannot transmit to each other, thus a contradiction. 
\end{proof}

\begin{lemma}
\label{lemma:hgt_log}
Under the assumption that the height of the MDCT $h_{orig} = O\left( \log (|V|)\right)$, 
the height of the DDCT constructed by LPR is $h_{ddct} = \Theta \left( \log (|V|)\right)$.
\end{lemma}

\begin{proof}
By the construction of the MDCT, the maximum distance of a node from the base
station is $h_{orig}R_{tx}$\footnote{Note that due to fading effects, the transmission range may not be a constant. 
However, there will always exist distances $R_0$ and $R_1$ such that if two nodes are within $R_0$ of each other,
they can transmit to each other with negligible loss, and if they are more than $R_1$ apart, they cannot exchange packets with each
other~\cite{aguayo:sigcomm,Govindan:Links}.  $R_0$ and $R_1$ may be much smaller and larger respectively than the actual
transmission range; 
replacing $R_{tx}$ by these constants appropriately allows the same argument 
to go through for a more general physical layer model.}. Using geometric arguments similar to the ones used in~\cite{theo:geometric}, it's easy to show 
that the set of nodes none of which can transmit to each other has a size of no more than $\frac{2 \pi}{\cos^{-1}
\left(1 - \frac{1}{2 h_{orig}^2} \right)} \leq \frac{2 \pi}{\cos^{-1}
\left(1 - \frac{1}{2 c^2 \log^2(|V|)} \right)}  \approx 2 \pi c \log(|V|)$, for some constant $c$, where the equality 
follows from the small angle approximation $\cos (x) \approx 1 - \frac{x^2}{2}$. 

Thus by Lemma~\ref{lemma:non-full}, there are no more than $2 \pi c \log(|V|)$
non-full heights.  At the same time, the number of full heights is $\Theta \left( \log (|V|)\right)$ by definition.
Hence 
$h_{ddct} = \Theta \left(\log (|V|)\right)$.
\end{proof}

\begin{theorem}
\label{thm:approx_aggr}
The approximation factor of LPR is $O \left( \log (|V|) \right)$. 
\end{theorem}
\begin{proof}
To derive the approximation factor, we compare the energy consumed in the DDCT constructed using LPR (given by Lemma~\ref{lemma:tree_energy}) to the lower bound on the optimal solution of P1 (given in Lemma~\ref{lemma:lb_no_comp}). First, we note that $|S| \geq \left( \frac{V}{n_{max}}\right)$ 
in the optimal solution and $|S| = c_1 \left( \frac{V}{n_{min}} \right)$ in the DDCT (as $h_{ddct} = \Theta \left(\log (|V|)\right)$) where $c_1$ is a positive constant,
$n_{max} = \mbox{max}_{v \in V} n_v$ and $n_{min} = \mbox{min}_{v \in V} n_v$.
Thus, the approximation factor is $ \leq \frac{\sum_{v \in V} \left( d_{ddct}(v) - 1 \right) + c_1 \left( \frac{V}{n_{min}}\right)}
{\sum_{v \in V} \left( d_0(v) - 1 \right) + \left( \frac{V}{n_{max}}\right)}$ $\leq \log (|V|)$, where $d_{ddct(v)}$ denotes the hop count of node $v$ in the DDCT. 
The final inequality holds because $h_{org} \leq c_2 \log (|V|)$ 
and $h_{ddct} = c_3 \log (|V|)$, 
for some positive constants $c_2$ and $c_3$. 
Hence the approximation factor is $O \left( \log (|V|)\right)$.
\end{proof}

\subsection{A Distributed Approximation Algorithm (DAA)}
\label{sec:sp_dist}

The approximation algorithm LPR is centralized as it requires solving
an LP globally. We now present a simpler, distributed algorithm with the same asymptotic approximation factor. 

The proof of Lemma~\ref{lemma:non-full} 
uses the following observation from LPR: at height $h$, if there exists a node $v$
with more than $n_v-1$ neighbors which are not yet a part of the tree, the algorithm will add $n_v-1$ children to it. 
Otherwise, all its neighbors not yet a part of the tree will be added as its children.

Using this intuition, we propose a modified Dijkstra's shortest path algorithm DAA in Figure~\ref{algo:dist_aggr}.
This algorithm satisfies the observation made in the previous paragraph, hence Lemma~\ref{lemma:non-full} holds,
and so do 
Lemma~\ref{lemma:hgt_log} and Theorem~\ref{thm:approx_aggr}.
Thus, the approximation factor for DAA is also $O \left( \log (|V|)\right)$.
The tree is built top down from the root with each node $v$ choosing its $n_v-1$ children arbitrarily. 
Hence, like any shortest path algorithm \cite{ctp} it can be built by message exchanges only between neighboring nodes. 
We will compare this modified Dijkstra's algorithm with LPR through simulation in Section~\ref{sec:simulations}.  

\begin{figure}[ht]
  \begin{center}
  \begin{ttfamily}
    \begin{footnotesize}
      \flushleft
       $NV = \{0\}$, $h_v = \infty, $ $\forall v \in V \backslash \{0\}$, $h_0 = 0$, $C_v = 0,$ $\forall v \in V$.\\
       ($C_v$ denotes the number of children of node $v$.) \\
       while $(NV != V)$ do \\
       \hspace*{0.2in} For each edge $e \in E$ such that $e$ connects \\
       \hspace*{0.2in} nodes $v \in NV$ and $v' \in V \backslash NV$ and $C_v < n_v - 1$
       \hspace*{0.4in} $h_v' = \mbox{min} \left(h_v', h_v + 1 \right)$ \\ 
       \hspace*{0.2in} $v_{min} = \mbox{argmin}_{v} \{ h_v \mid \forall v \in V \backslash NV \}$ \\
       \hspace*{0.2in} Add $v_{min}$ to $NV$. \\
       \hspace*{0.2in} Let the parent of $v_{min}$ be $v_{parent}$. Update \\
       \hspace*{0.2in} $C_{v_{parent}} = C_{v_{parent}} + 1$\\
       \hspace*{0.2in} Set $h_v = \infty,$ $\forall v \in V \backslash NV$ \\
\end{footnotesize}
\end{ttfamily}
\end{center}
\vspace{-0.15in}
  \caption{Algorithm DAA: Modified Dijkstra's algorithm for P3.}
  \label{algo:dist_aggr}
\end{figure}

\section{Discussion} 
\label{sec:disc_aggr}

In this section we discuss a number of ways to relax the assumptions used earlier, 
as well as the applicability of distributed SVD computation in practice. 


The energy model presented in Section~\ref{sec:model} 
is rather simplistic; it does not capture energy expended in overhearing etc. 
However, as long as the energy model is a linear function of the amount of bits  transmitted per node (most energy consumption models fit this
characterization), 
the proposed algorithms can be directly applied without any change in their optimality or approximation factors. 


The model and algorithms presented here can also be easily extended to include additional constraints, including accuracy and storage.  
%
In our decentralized SVD computation, the eigenvectors are determined by linearly combining those computed locally at different sensor nodes. 
If the sensors are noiseless, then the eigenvectors computed using this decomposition will exactly 
match the actual eigenvectors.  However, the presence of noise in the sensed values can lead to errors in the computation~\cite{andy:market}.  
This is because in a centralized implementation, a least-squares effect minimizes the error due to noise across all eigenvectors,
whereas the decentralized implementation allows this error to accumulate through each combination of locally computed eigenvectors. 

The larger the number of FFT's being combined at each sensor node, 
the smaller this error.
Hence a desired accuracy will impose a constraint on the minimum cluster size $|N_s|, s \in S$.  
%
This is the opposite of the delay constraint, and incorporating it in our models is quite straightforward. 
Denote this constraint by $n_a$, i.e.,  $|N_s|\geq n_a, \forall s \in S$. 
Then in ILP\_P1, the following constraint is added: $\sum_{i \in V} x_{ij} \geq n_a x_{jj}, \forall j \in V$.
Similarly, in ILP\_P3, we add (1) $\sum_{e \in I_v} x_e \geq (n_a - 1) l_v, \forall v \in V$, where variable $l_v \in \{0, 1\}$ is
set to $1$ if $v$ is a non-leaf node, and (2) $\sum_{e \in I_v} x_e/|V| \leq l_v \leq \sum_{e \in I_v} x_e, \forall v \in V$, to 
ensure that $l_v$ is set $1$ only if $v$ is a non-leaf node.  
Finally, the two approximation algorithms, LPR and DAA,  can both be easily modified to maintain the number of children of each 
node in the data collection to be greater than $n_a - 1$.  
%
The effect of an added accuracy constraint will be examined in numerical studies presented in Section \ref{sec:simulations}. 

As the number of FFT's being computed at a node increases, not only the delay but also
the storage required increases~\cite{andy:market}.  A storage constraint acts in a way very similar to the delay constraint: 
it essentially bounds the maximum number of FFT's that can be combined at a sensor. 
Therefore to incorporate this constraint we simply need to upper bound the value of $|N_s|$ to be the lesser of the two, which 
results in an identical problem. 

While our discussion has centered solely on the computational task of SVD, our approach is more generally applicable.  Once a computational task is represented as a set of operators with associated input and output dependencies, one can use a very similar approach to seek the optimal communication structure, i.e., on which node to place which operator and along what path to send input to that node, etc.  The resulting math program will in general be  problem specific, but the solution philosophy is common; 
see Appendix~\ref{sec:annealing} for a similar approach to decomposing the global operation of simulated annealing over a network of sensor nodes.


\comment{
While in-network SVD computation is clearly more energy efficient, one might question how much is really gained if this operation is performed infrequently (note that it can certainly be configured to occur continuously), as one does not in general expect mode shapes of a structure to change rapidly over time; 
the physical degradation process is on a very different time scale than computation and communication.  For a wireless sensor network instrumented on a structure, say a bridge, this task could be reasonably scheduled several times a day, each lasting on the order of minutes (the actual duration of this sensing cycle depends on the size of the FFT and the computation capacity of the sensors).  We observe that the saving in each such operation is indeed significant (see results in Section \ref{sec:simulations}), and the accumulated effect are undoubtedly beneficial for a sensor network to have a lifetime on the order of months or years. 
} 

While the proposed SVD computation can run continuously as a stream process, in practice it suffices to schedule it several times a day, each lasting on the order of minutes (the actual duration of the sensing cycle depends on the size of the FFT and the computation capacity of the sensors), as one does not in general expect mode shapes of a structure to change rapidly over time.  Even though the task is performed infrequently, the saving in each operation is indeed significant (see results in Section \ref{sec:simulations}), and the accumulated effect are undoubtedly beneficial for a sensor network to have a lifetime on the order of months or years. 

One weakness common to most in-network processing methods is that they typically deliver {\em summaries} or {\em features} of data rather than raw data itself; thus we potentially lose the ability to store and post-analyze the data (e.g., for an entirely different purpose than originally intended).  In this sense, this type of operation is most advantageous when used in a real-time setting concerning instantaneous detection and diagnosis.  For instance, a human inspector can use this approach (i.e., activate this SVD operation) to quickly check the mode shapes of a structure before deciding whether and what more (manual) inspection is needed.  
 

\section{Simulation and Experimentation}
\label{sec:simulations}

\begin{figure*}[ht]
\centerline{\subfigure[]{\includegraphics[width=4.5cm]{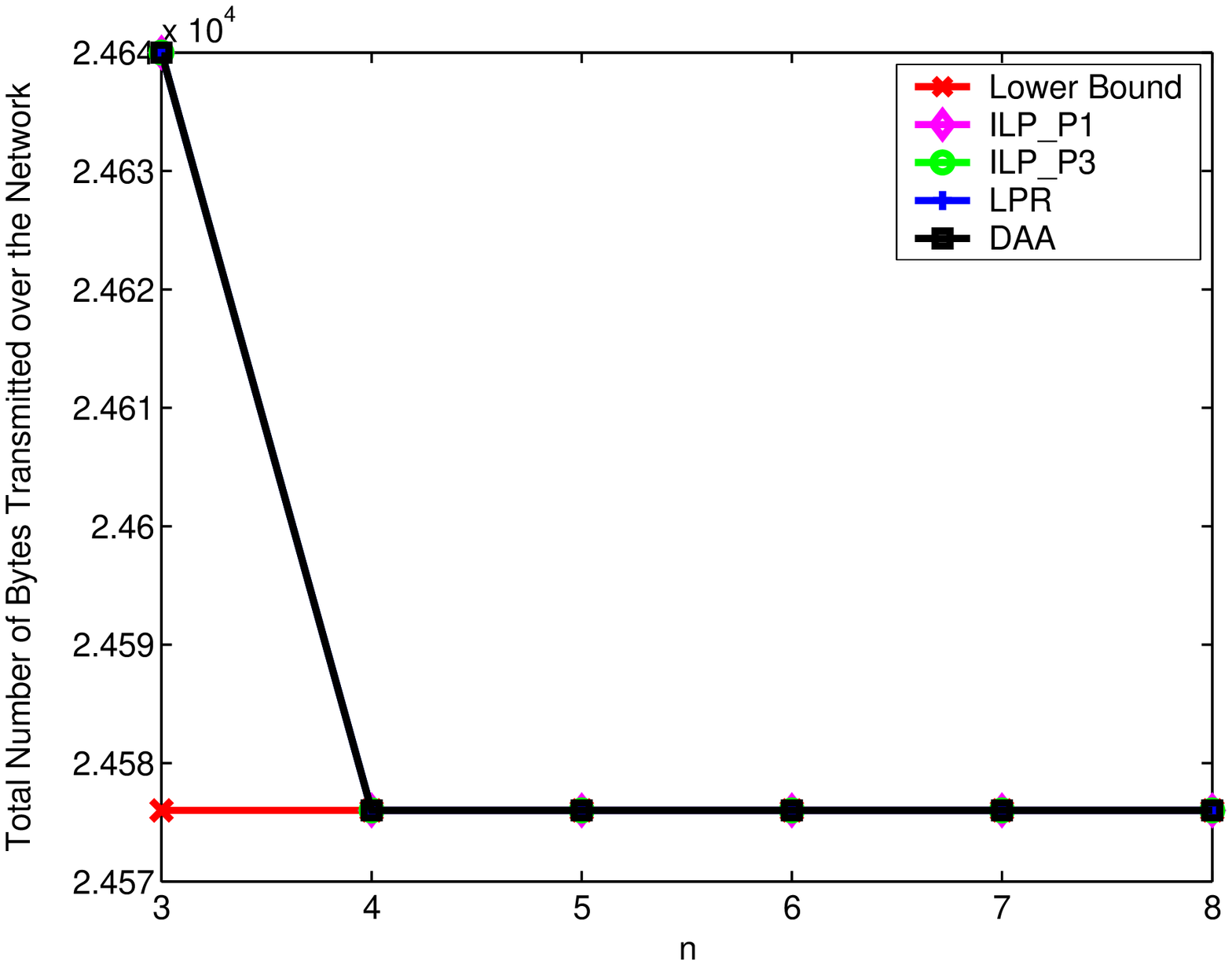}
\label{fig:svd1}}
\hfil
\subfigure[]{\includegraphics[width=4.5cm]{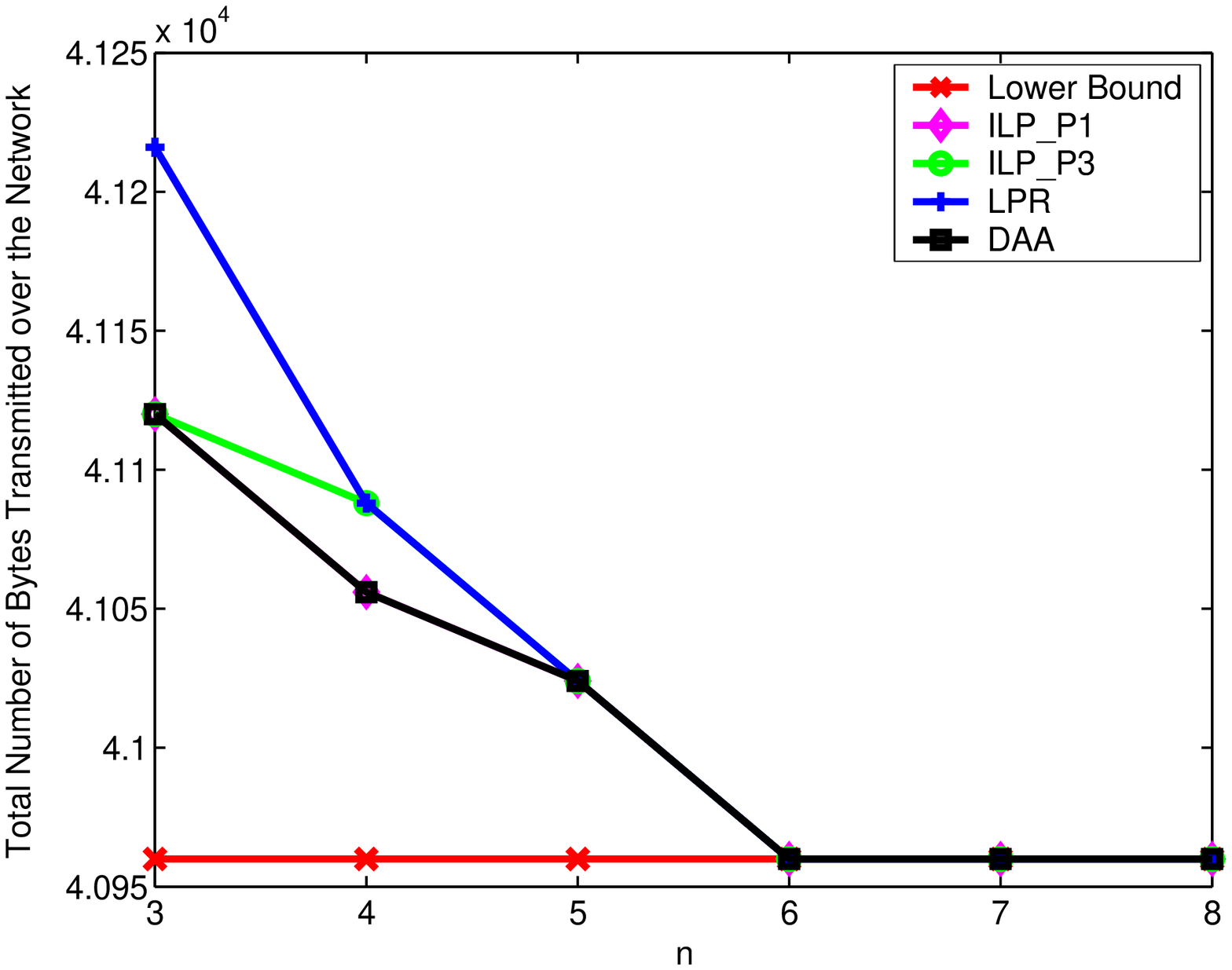}
\label{fig:svd2}}
\hfil
\subfigure[]{\includegraphics[width=4.5cm]{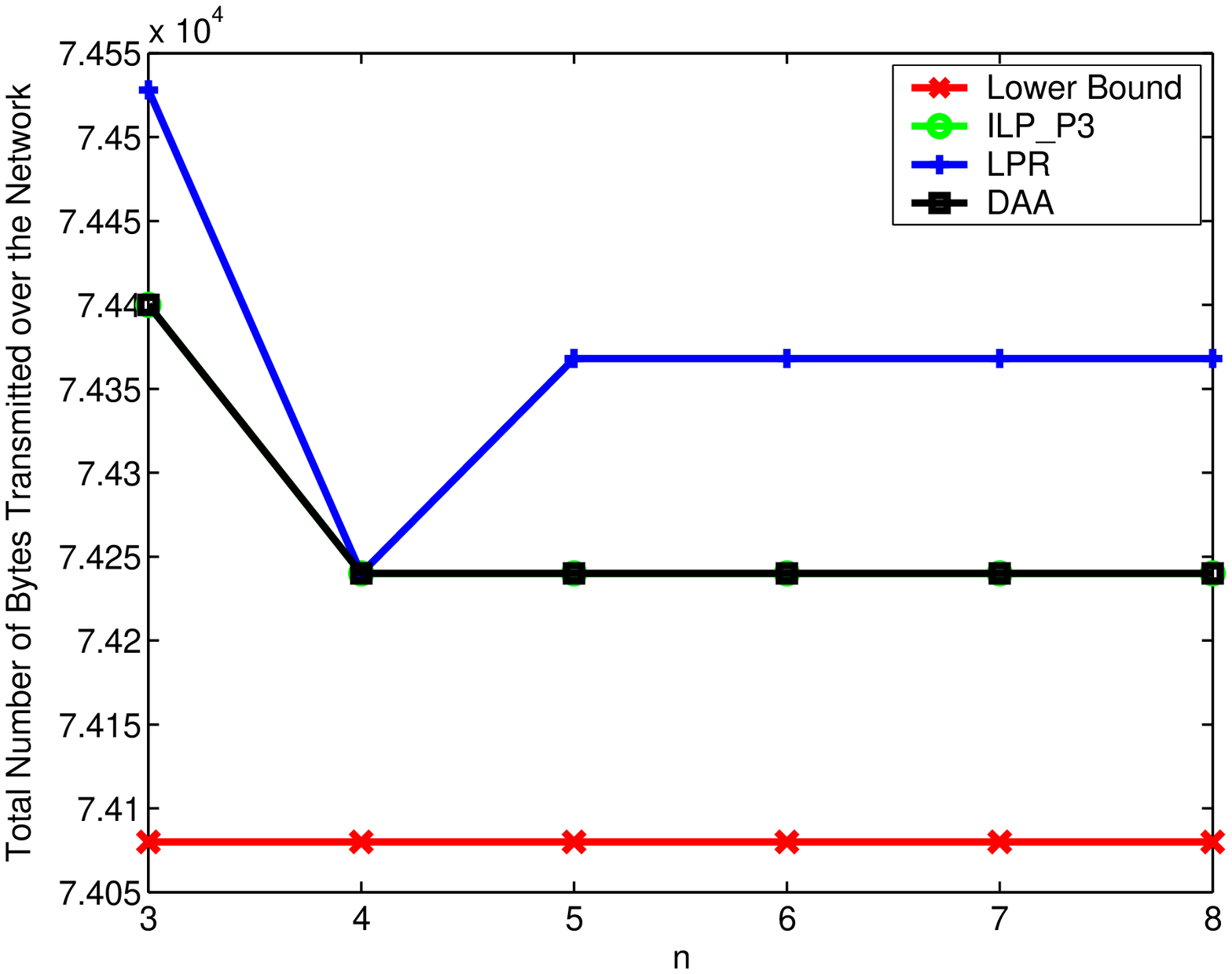}
\label{fig:svd3}}}
\centerline{\subfigure[]{\includegraphics[width=4.5cm]{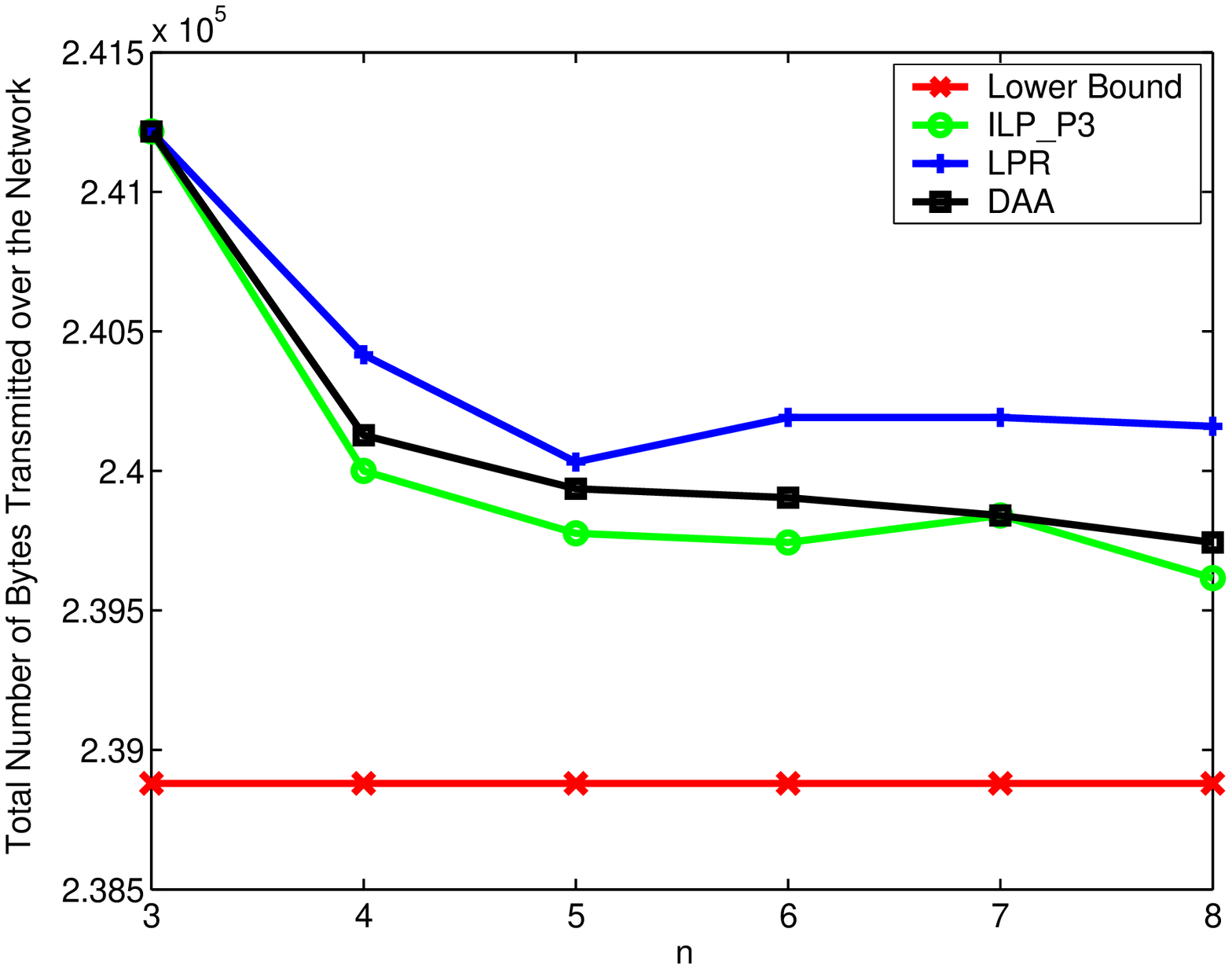}
\label{fig:svd4}}
\hfil
\subfigure[]{\includegraphics[width=4.5cm]{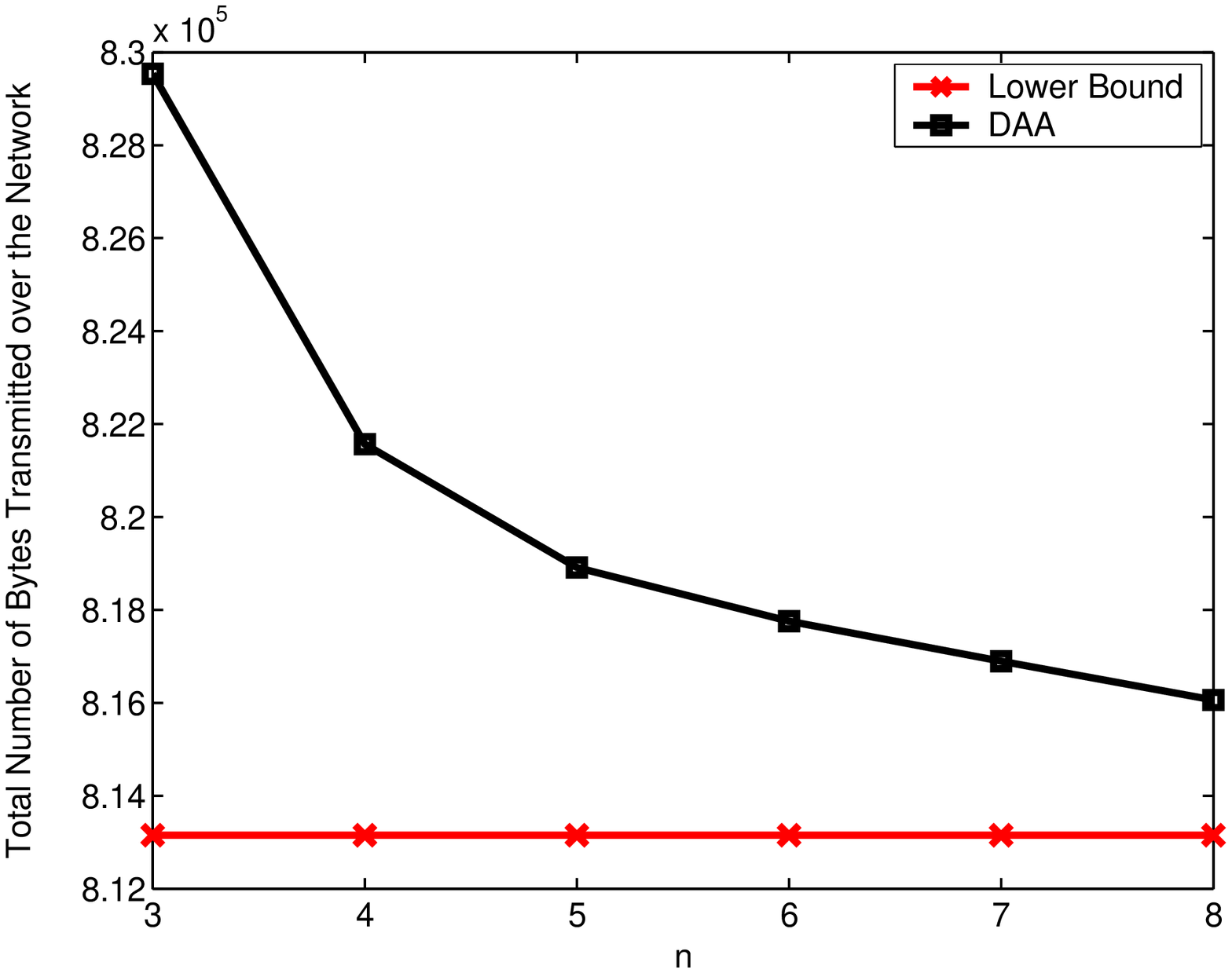}
\label{fig:svd5}}
\hfil
\subfigure[]{\includegraphics[width=4.5cm]{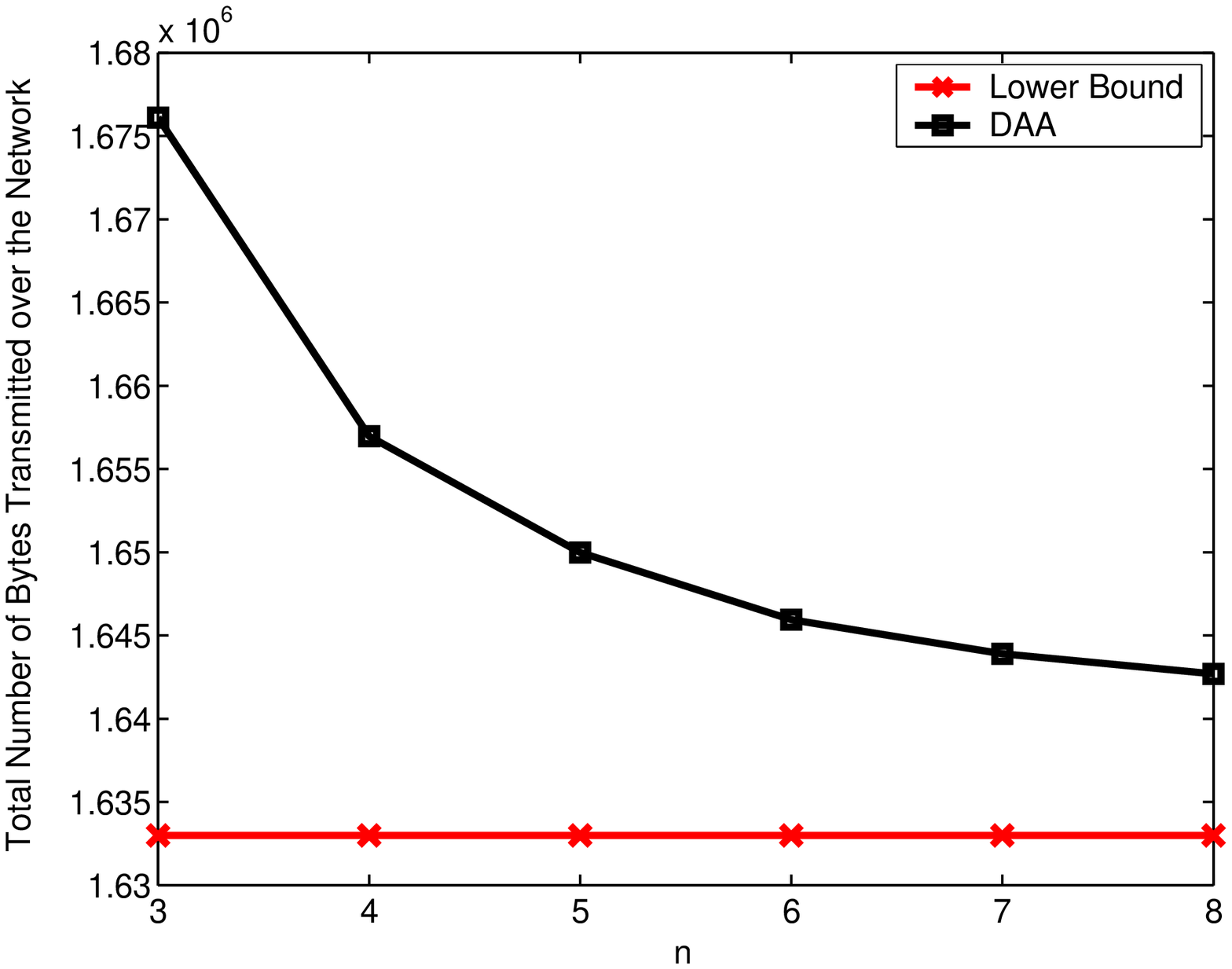}
\label{fig:svd6}}}
\centerline{\subfigure[]{\includegraphics[width=4.5cm]{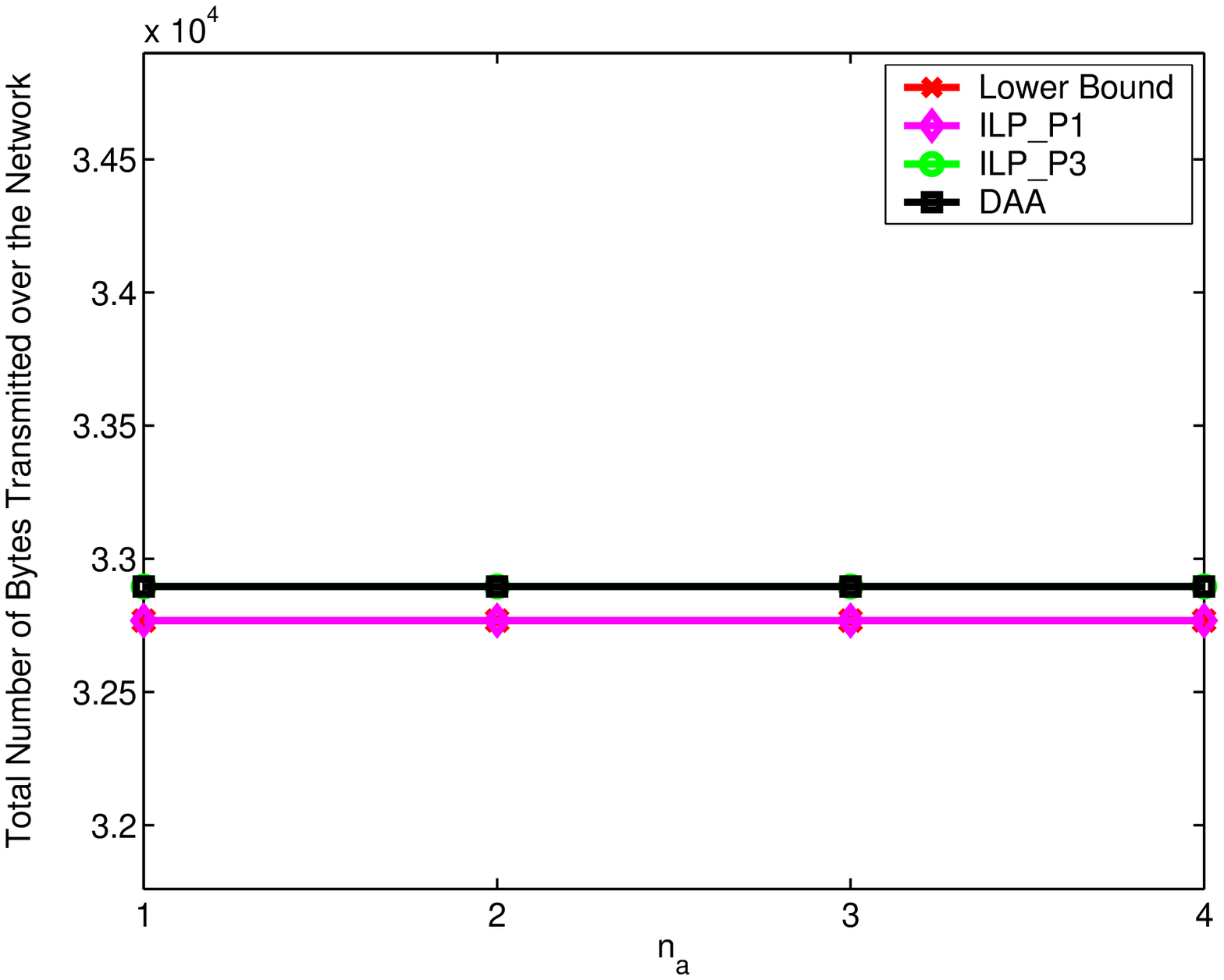}
\label{fig:accuracy1}}
\hfil
\subfigure[]{\includegraphics[width=4.5cm]{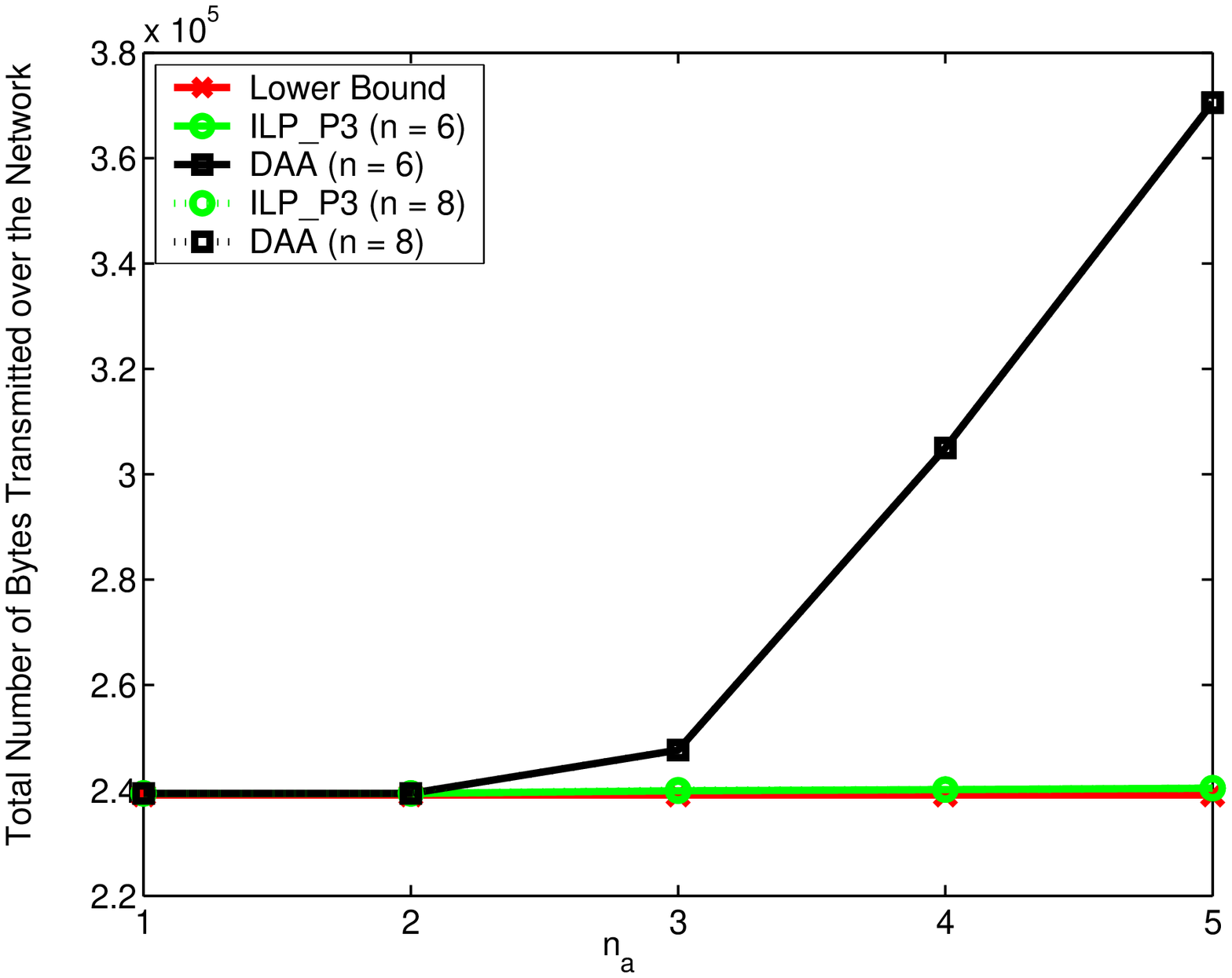}
\label{fig:accuracy2}}
\hfil
\subfigure[]{\includegraphics[width=4.5cm]{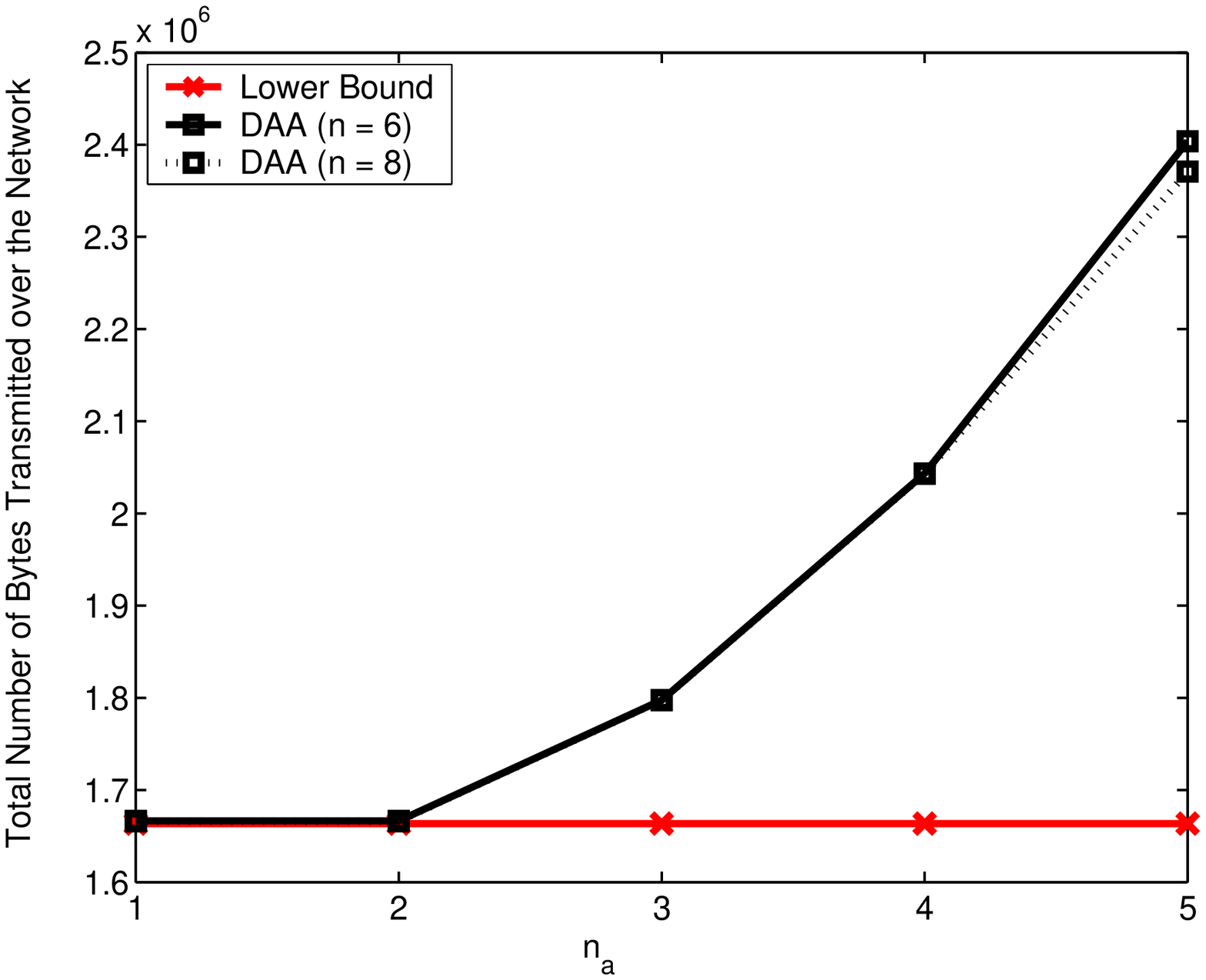}
\label{fig:accuracy3}}}
\caption{
Simulation Results. (a) $|V|=4$ ($24576$). (b) $|V|=6$ ($40960$). (c) $|V|=10$ ($163840$). (d) $|V|=30$ ($573440$). (e) $|V|=100$ ($1359872$). 
(f) $|V|=200$ ($2342912$). The number in brackets denotes the number of bytes transmitted in the network without in-network computation.
Simulation Results with an accuracy constraint. (g) $|V|=5, n=5$. (h) $|V|=30$. (i) $|V|=200$. 
}
\end{figure*}

We use both simulation and experimentation on a real sensor platform to evaluate the performance of the proposed algorithms. 
For simulation we use CPLEX~\cite{cplex} to solve the ILPs, and 
all simulations are done on topologies generated by randomly distributing nodes in an area of $50 \times 50 m^2$ and assuming the transmission range to be $30 m$. 
For the SVD computation, we use $R=8192$ bytes and $r=32$ bytes~\cite{andy:svd}. 
We also assume that the computational delay constraint is the same for all nodes and $n_v = n, \forall v \in V$.

We first examine the effect of delay constraint $n$ on energy consumption. 
Figures~\ref{fig:svd1} and~\ref{fig:svd2} compare the number of bytes transmitted under the lower bound (Lemma~\ref{lemma:lb_no_comp}), using the optimal communication structures derived by solving ILP\_P1 (Section~\ref{ilp:original}), 
and using the three approximation algorithms ILP\_P3 (Section~\ref{sec:dtree-ilp}), LPR (Figure~\ref{algo:with_aggr}), and DAA (Figure~\ref{algo:dist_aggr}), for different values of $n$, with $|V|=4$ and $|V|=6$ respectively. 

We observe that the approximation algorithms perform very close to the optimal. 
It takes more than one hour of computation to solve the ILP\_P1 for $|V| > 6$ on a 2.99 GHz machine with
4 GB of RAM. Hence for larger values of $|V|$ we only compare the three approximation algorithms against the lower bound, shown  
in Figures~\ref{fig:svd3} and~\ref{fig:svd4}. We note that 
(i) all approximation algorithms are within $3\%$ of the optimal,
and (ii) DAA outperforms LPR.
These results also demonstrate the advantage of using the ILP\_P3 over ILP\_P1; it runs much faster  
and converges within an hour up to $|V| = 40$. 

For even larger values of $|V|$, we compare the performance of DAA (as it consistently
outperforms LPR) against the lower bound in Figures~\ref{fig:svd5} and~\ref{fig:svd6}. 
We observe that it is always within $3\%$ of the optimal. These results clearly demonstrate the advantage of in-network
computation as the number of bytes transmitted over the network are reduced by more than half. 
Finally, Figures~\ref{fig:svd5} and~\ref{fig:svd6} also show the trade-off between communication energy and computational delay.  The more delay allowed per node (larger the value of $n$), the smaller the energy consumed in the network. 

In Figures~\ref{fig:accuracy1}-\ref{fig:accuracy3}, we compare the performance of different approximation schemes after incorporating an accuracy 
constraint in the formulation for different values of $|V|$, $n$ and $n_a$. In this scenario, we observe that ILP\_P3 
yields results within $5\%$ of the optimal while DAA yields values within
$51\%$ of the optimal. And the advantage of using a better centralized algorithm becomes more pronounced as the value of $n_a$ increases as
any sub-optimal local decision in this scenario leads to an extra transmission of a FFT ($R$ bits) and not just an eigenvector ($r$ bits).

\begin{figure*}[ht]
\centerline{\subfigure[]{\includegraphics[width=5.0cm]{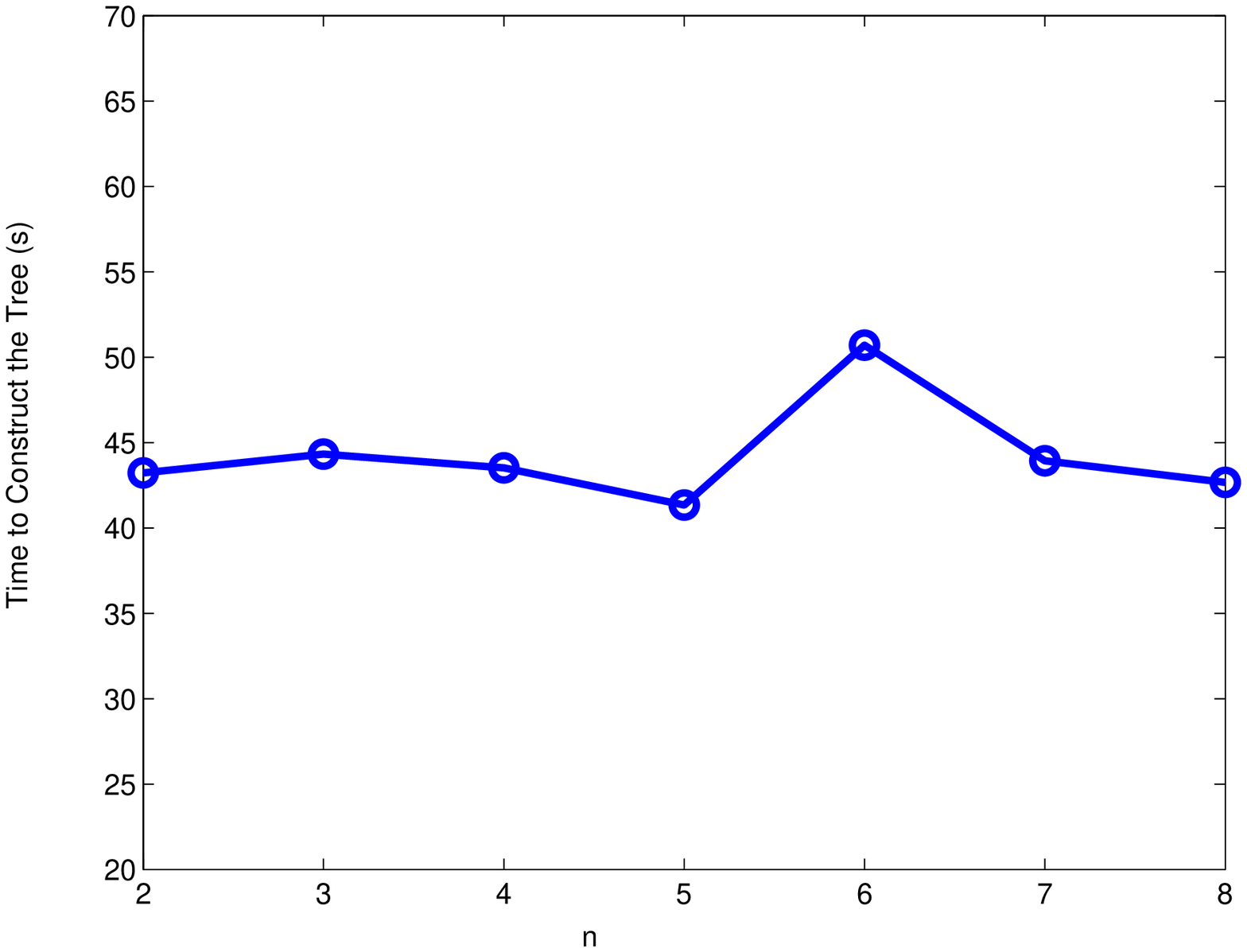}
\label{fig:time_d}}
\hfil
\subfigure[]{\includegraphics[width=5.0cm]{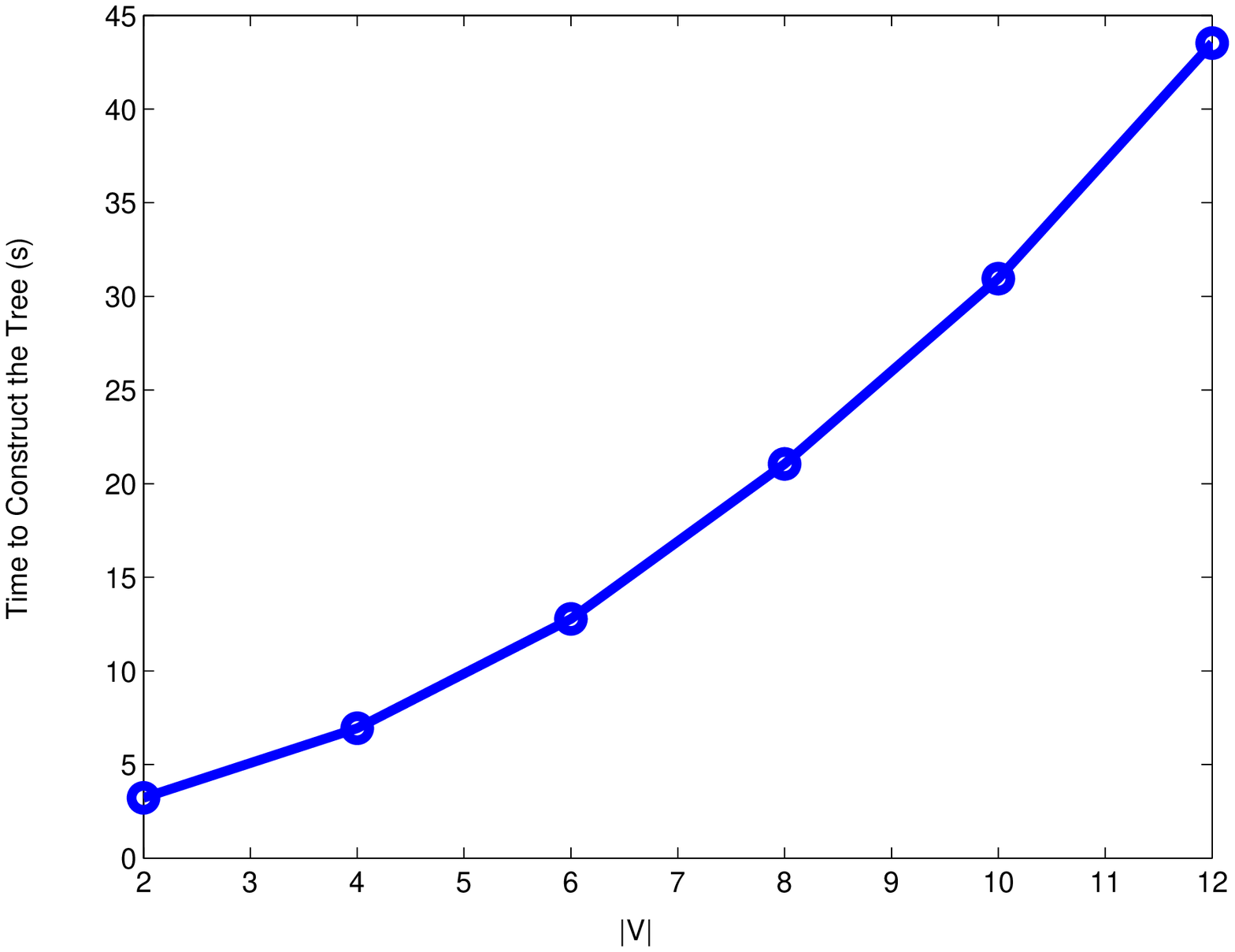}
\label{fig:time_v}}
\hfil
\subfigure[]{\includegraphics[width=5.0cm]{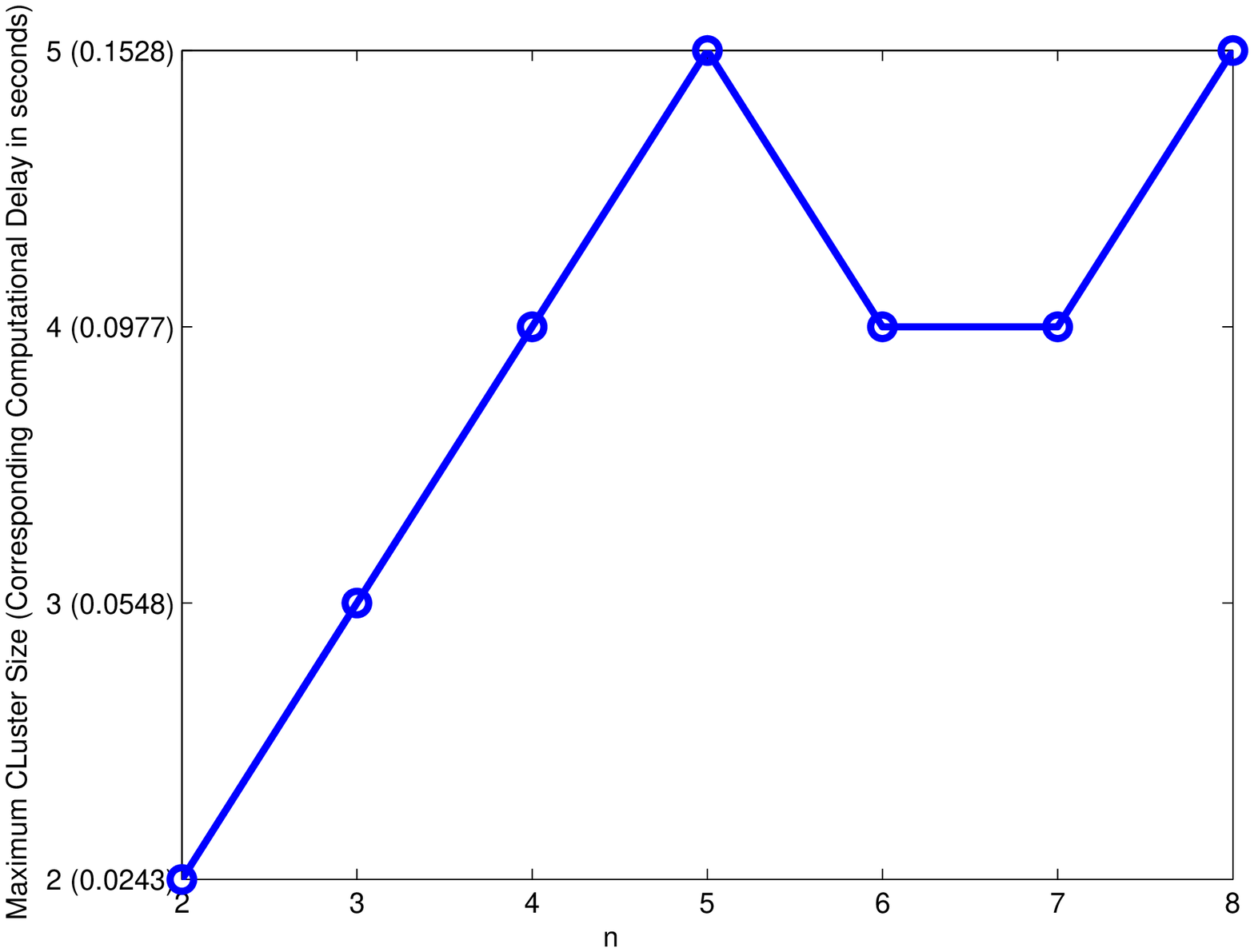}
\label{fig:sensing_d}}}
\caption{(a) Time to construct the tree vs $n$; $|V|=12$. (b) Time to construct the tree vs $|V|$; $n=4$. (c) Maximum cluster size (corresponding computational
delay in seconds) vs $n$; $|V|=12$.}
\end{figure*}

We next evaluate the performance of DAA on a real sensor platform, the Narada sensing unit developed at the University of Michigan~\cite{narada}.
This wireless device is powered by an Atmel ATmega 128 microprocessor. It is supplemented by 128 KB of external SRAM and utilizes the 4-channel,
16-bit ADS8341 ADC for data acquisition.  Narada's wireless communication interface consists of Chipcon CC2420 IEEE 802.15.4 compliant 
transceiver, which makes it an extremely versatile unit for developing large-scale WSNs. This prototype is powered by a constant DC supply voltage 
between 7 and 9 volts, and has an operational life expectancy of approximately 48 hours with 6 AA batteries, given constant communication
and data analysis demands.  

We use a testbed of 12 Narada sensor nodes deployed in a corridor in the Electrical Engineering and Computer Science building at the 
University of Michigan. Each Narada wireless sensor is programmed with DAA algorithm, and asked to autonomously form computational clusters with varying values of $n$. The root of the tree is randomly selected in each experiment. 
In a manner similar to~\cite{andy:market}, the weight of an edge $e$ is set to 
$w_e = \frac{1-p_{CF}}{1+e^{-0.4\left( 40+RSSI\right)}}$, where $RSSI$ is the radio signal strength indicator reported by the radio
and $p_{CF}$ is the probability that a communication link with perfect RSSI fails due to unforeseen circumstances and is set to $0.1$
for the Narada platform.

The objective of our experiment is 
to study the time it takes to construct the data collection tree using DAA, 
as well as the cluster sizes and the corresponding sensing cycles as a function of $n$ in a real-world setting. 
Figures~\ref{fig:time_d} and~\ref{fig:time_v} plot the time it takes to construct the tree as a function of $n$ and $|V|$ respectively. 
We see that this time only depends on the size of the network.
Figure~\ref{fig:sensing_d} plots the maximum cluster size as well as the maximum computational delay for the corresponding 
cluster size 
as a function of $n$.  
Figure~\ref{fig:dct_d} shows the data collection trees constructed for $n=3$ and $n=5$, respectively.
To summarize, the implementation and the experimental results verify the feasibility of DAA in a real SHM sensor network.
\begin{figure}[ht]
 \centering
\includegraphics[width=3.2in]{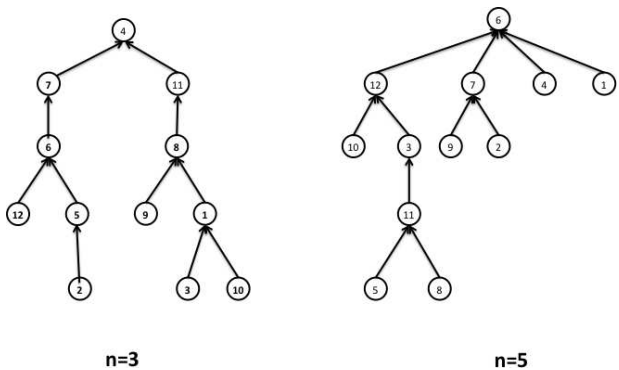}
  \caption{Data collection trees for $n=3$ and $n=5$.}
\label{fig:dct_d}
\vspace{-0.15in}
\end{figure}

\section{Conclusions}
\label{sec:conclusions}
This paper studies the problem of networked computation within the context of wireless sensor networks used for structural health monitoring.  
It presents centralized ILPs and distributed approximation algorithms to derive optimal communication
structures for the distributed computation of SVD.  
Both simulations and implementations are used to evaluate their performance. 
Our results demonstrate the advantage of in-network computation as it significantly reduces the amount of data transmitted over the 
network.


%

\small
\bibliographystyle{IEEEtran}
\vspace{0.4in}
\bibliography{my} 

\begin{thebibliography}{10}
\providecommand{\url}[1]{#1}
\csname url@samestyle\endcsname
\providecommand{\newblock}{\relax}
\providecommand{\bibinfo}[2]{#2}
\providecommand{\BIBentrySTDinterwordspacing}{\spaceskip=0pt\relax}
\providecommand{\BIBentryALTinterwordstretchfactor}{4}
\providecommand{\BIBentryALTinterwordspacing}{\spaceskip=\fontdimen2\font plus
\BIBentryALTinterwordstretchfactor\fontdimen3\font minus
  \fontdimen4\font\relax}
\providecommand{\BIBforeignlanguage}[2]{{%
\expandafter\ifx\csname l@#1\endcsname\relax
\typeout{** WARNING: IEEEtran.bst: No hyphenation pattern has been}%
\typeout{** loaded for the language `#1'. Using the pattern for}%
\typeout{** the default language instead.}%
\else
\language=\csname l@#1\endcsname
\fi
#2}}
\providecommand{\BIBdecl}{\relax}
\BIBdecl

\bibitem{tag}
S.~Madden, M.~Franklin, J.~Hellerstein, and W.~Hong, ``{TAG: a Tiny AGgregation
  service for ad-hoc sensor networks},'' in \emph{Proc. of OSDI}, 2002.

\bibitem{median:podc}
B.~Patt-Shamir, ``{A note on efficient aggregate queries in sensor networks},''
  in \emph{Proc. of ACM PODC}, 2004.

\bibitem{max:pods}
S.~Kashyap, S.~Deb, K.~Naidu, and R.~Rastogi, ``{Efficient gossip-based
  aggregate computation},'' in \emph{Proc. of ACM PODS}, 2006.

\bibitem{wang:convergcast}
X.~Li, Y.~Wang, and Y.~Wang, ``Complexity of convergecast and data selection
  for wireless sensor networks,'' Illinois Institute of Technology, Tech. Rep.,
  2009.

\bibitem{kumar:boolean}
H.~Kowshik and P.~Kumar, ``{Optimal strategies for computing symmetric Boolean
  functions in collocated networks},'' in \emph{Proc. of IEEE Information
  Theory Workshop}, 2010.

\bibitem{kumar:scaling}
A.~Giridhar and P.~R. Kumar, ``{Towards a Theory of In-Network Computation in
  Wireless Sensor Networks},'' \emph{IEEE Communications Magazine}, vol.~44,
  no.~4, pp. 98--107, 2006.

\bibitem{kumar:function}
H.~Kowshik and P.~Kumar, ``{Zero-error function computation in sensor
  networks},'' in \emph{Proc. of CDC}, 2009.

\bibitem{medard:computing}
S.~Feizi and M.~Medard4, ``{When do only sources need to compute? On functional
  compression in tree networks},'' in \emph{Proc. of the Allerton Conference on
  Control, Communications and Computing}, 2009.

\bibitem{srikant:function}
L.~Ying, R.~Srikant, and G.~Dullerud, ``{Distributed Symmetric Function
  Computation in Noisy Wireless Sensor Networks},'' \emph{IEEE Transactions on
  Information Theor}, vol.~53, no.~12, 2007.

\bibitem{mode1}
S.~Fang and R.~Perera, ``{Power mode shapes for early damage detection in
  linear structures},'' \emph{Journal of Sound and Vibration}, vol. 324, no.
  1-2, pp. 40--56, 2009.

\bibitem{mode2}
Z.~Ismail, H.~Razak, and A.~Rahman, ``{Determination of damage location in RC
  beams using mode shape derivatives},'' \emph{Engineering Structures},
  vol.~28, no.~11, pp. 1566--1573, 2006.

\bibitem{mode3}
E.~Clayton, B.~Koh, G.~Xing, C.~Fok, S.~Dyke, and C.~Lu, ``{Damage detection
  and correlation-based localization using wireless mote sensors},'' in
  \emph{IEEE International Symposium on Intelligent Control}, 2005.

\bibitem{andy:svd}
A.~Zimmerman, M.Shiraishi, R.~Swartz, and J.~Lynch, ``{Automated Modal
  Parameter Estimation by Parallel Processing within Wireless Monitoring
  Systems},'' \emph{Journal of Infrastructure Systems}, vol.~14, no.~1, pp.
  102--113, 2008.

\bibitem{svd1}
D.~Friedlander, C.~Griffin, N.~Jacobson, S.~Phoha, and R.~Brooks, ``{Dynamic
  agent classification and tracking using an ad hoc mobile acoustic sensor
  network},'' \emph{EURASIP Journal on Applied Signal Processing}, vol. 2003,
  no.~4, pp. 371--377, 2003.

\bibitem{svd2}
J.~Gupchup, R.~Burns, A.~Terzis, and A.~Szalay, ``{Model-Based Event Detection
  in Wireless Sensor Networks},'' in \emph{Proc. of Workshop on Data Sharing
  and Interoperability on the World-Wide Sensor Web}, 2007.

\bibitem{svd3}
L.~Balzano and R.~Nowak, ``{Blind Calibration of Sensor Networks},'' in
  \emph{Proc. of IPSN}, 2007.

\bibitem{svd4}
G.~Derveaux, G.~Papanicolaou, and C.~Tsogka, ``{Time reversal imaging for
  sensor networks with optimal compensation in time},'' \emph{The Journal of
  the Acoustical Society of America}, vol. 121, no.~4, pp. 2071--2085, 2007.

\bibitem{svd5}
V.~Raykar, I.~Kozintsev, and R.~Lienhart, ``{Position Calibration of
  Microphones and Loudspeakers in Distributed Computing Platforms},''
  \emph{IEEE Transactions on Speech and Audio Processing}, vol.~13, no.~1,
  2005.

\bibitem{bazarul:ton}
R.~Cristescu, B.~Beferull-Lozano, M.~Vetterli, and R.~Wattenhofer, ``{Network
  correlated data gathering with explicit communication: NP-completeness and
  algorithms},'' \emph{IEEE/ACM Transactions on Networking}, vol.~14, no.~1,
  pp. 41--54, 2006.

\bibitem{pp2}
D.~Ewins, \emph{{Modal testing: Theory and practice}}, 2nd~ed.\hskip 1em plus
  0.5em minus 0.4em\relax Research Studies Press Ltd.

\bibitem{pp1}
O.~Salawu, ``{Detection of structural damage through changes in frequency: A
  review},'' \emph{Engineering Structures}, vol.~19, no.~9, pp. 718--723, 1997.

\bibitem{fdd}
R.~Brincker, L.~Zhang, and P.~Andersen, ``{Modal identification of output-only
  systems using frequence domain decomposition},'' \emph{Smart Materials and
  Structures}, vol.~10, no.~3, pp. 441--445, 2001.

\bibitem{sivakumar:secon}
Y.~Zhu, K.~Sundaresan, and R.~Sivakumar, ``{Practical Limits on Achievable
  Energy Improvements and Useable Delay Tolerance in Correlation Aware Data
  Gathering in Wireless Sensor Networks},'' in \emph{Proc. of IEEE SECON},
  2005.

\bibitem{pattem:ipsn}
S.~Pattem, B.~Krishnmachari, and R.~Govindan, ``{The Impact of Spatial
  Correlation on Routing with Compression in Wireless Sensor Networks},'' in
  \emph{Proc. of IPSN}, 2004.

\bibitem{leach}
W.~Heinzelman, A.~Chandrakasan, and H.~Balakrishnan, ``{An application-specific
  protocol architecture for wireless microsensor networks},'' \emph{IEEE
  Transactions on Wireless Communications}, vol.~1, no.~4, pp. 660--670, 2002.

\bibitem{nsdi:wishbone}
R.~Newton, S.~Toledo, L.~Girod, H.~Balakrishnan, and S.~Madden, ``{Wishbone:
  Profile-based Partitioning for Sensornet Applications},'' in \emph{Proc. of
  NSDI}, 2009.

\bibitem{tree:apx}
B.~Brinkman and M.~Helmick, ``Degree-constrained minimum latency trees are
  apx-hard,'' Miami University, Tech. Rep. 2008-03-25, 2008.

\bibitem{helmick:multicast}
M.~Helmick and F.~Annexstein, ``{Depth-Latency Tradeoffs in Multicast Tree
  Algorithms},'' in \emph{Proc. of The IEEE 21st International Conference on
  Advanced Information Networking and Applications}, 2007.

\bibitem{aguayo:sigcomm}
D.~Aguuayo, J.~Bicket, S.~Biswas, G.~Judd, and R.~Morris, ``{Link-Level
  Measurements from an 802.11b Mesh Network},'' in \emph{Proc. of ACM SIGCOMM},
  2004.

\bibitem{Govindan:Links}
J.~Zhao and R.~Govindan, ``{Understanding packet delivery performance in dense
  wireless sensor networks},'' in \emph{Proc. of ACM SenSys}, 2003.

\bibitem{theo:geometric}
W.~Wu, H.~Du, X.~Jia, Y.~Li, and S.-H. Huang, ``{Minimum Connected Dominating
  Sets and Maximal Independent Sets in Unit Disk Graphs},'' \emph{Theoretical
  Computer Science}, vol. 352, no.~1, pp. 1--7, 2006.

\bibitem{ctp}
O.~Gnawali, R.~Fonseca, K.~Jamieson, D.~Moss, and P.~Levis, ``{Collection Tree
  Protocol},'' in \emph{Proc. of ACM SenSys}, 2009.

\bibitem{andy:market}
A.~Zimmerman and J.~Lynch, ``{Market-based frequency domain decomposition for
  automated mode shape estimation in wireless sensor networks},'' \emph{Journal
  of Structural Control and Health Monitoring}, 2010.

\bibitem{cplex}
``{ILOG CPLEX 11.0},'' {http://www.ilog.com/products/cplex}.

\bibitem{narada}
A.~Swartz, D.~Jung, J.~Lynch, Y.~Wang, D.~Shi, and M.~Flynn, ``{Design of a
  Wireless Sensor for Scalable Distributed In-Network Computation in a
  Structural Health Monitoring System},'' in \emph{Proc. of the 5th
  International Workshop on Structural Health Monitoring}, 2005.

\bibitem{golub:svd}
G.H.Golub and C.Reinsch, ``{Singular Value Decomposition and Least Squares
  Solution},'' \emph{Numerical Mathematics}, vol.~14, pp. 403--420, 1970.

\bibitem{np:book}
M.~Garey and D.~Johnson, \emph{{Computers and Intractability: A Guide to the
  Theory of NP-Completeness}}.\hskip 1em plus 0.5em minus 0.4em\relax W. H.
  Freeman and Co.

\bibitem{model_update1}
S.~Doebling, C.~Farrar, and M.~Prime, ``{A Summary Review of Vibration-Based
  Damage Identification Methods},'' \emph{The Shock and Vibration Digest},
  vol.~30, no.~2, pp. 91--105, 1998.

\bibitem{model_update2}
A.~Teughels, J.~Maeck, and G.~Roeck, ``{Damage assessment by FE model updating
  using damage functions},'' \emph{Computers and Structures}, vol.~80, no.~25,
  pp. 1869--1879, 2002.

\bibitem{model_update3}
J.~Mottershead and M.~Friswell, ``{Model Updating In Structural Dynamics: A
  Survey},'' \emph{Journal of Sound and Vibration}, vol. 167, no.~2, pp.
  347--375, 1993.

\bibitem{sa1}
R.~Levin and N.~Lieven, ``{Dynamic finite element model updating using
  simulated annealing and genetic algorithms},'' \emph{Mechanical Systems and
  Signal Processing}, vol.~12, no.~1, pp. 91--120, 1998.

\bibitem{sa2}
N.~Metropolis, A.~Rosenbluth, M.~Rosenbluth, A.~Teller, and E.~Teller,
  ``{Equation of State Calculations by Fast Computing Machines},'' \emph{The
  Journal of Chemical Physics}, vol.~21, no.~6, pp. 1087--1092, 1953.

\bibitem{sa3}
D.~Greening, ``{Parallel simulated annealing techniques},'' \emph{Physica D},
  vol.~42, pp. 293--306, 1990.

\bibitem{andy:sa}
A.~Zimmerman and J.~Lynch, ``{A Parallel Simulated Annealing Architecture for
  Model Updating in Wireless Sensor Networks},'' \emph{EEE Sensors Journal},
  vol.~9, no.~11, pp. 1503--1510, 2009.

\end{thebibliography}
\normalsize

\appendices
\section{A Brief Overview of Singular Value Decomposition}
\label{appendix2}
Let $A$ be a real $m \times n$ matrix with $m \geq n$. Then, the singular value decomposition (SVD) factors $A$ as follows:
$A = U \Sigma V^T$ where $U^T U = V^T V = V V^T = I_n$ and $\Sigma = diag(\sigma_1, \sigma_2, \ldots , \sigma_n)$. 
The matrix $U$ consists of $n$ orthonormalized eigenvectors associated with the $n$ largest eigenvalues of $A A^T$,
and the matrix $V$ consists of the orthonormalized eigenvectors of $A^T A$. The diagonal elements of $\Sigma$ are
the non-negative square roots of the eigenvalues of $A^T A$. We shall assume that $\sigma_1 \geq \sigma_2 \geq \ldots \geq \sigma_n$. 
SVD comprises of two steps~\cite{golub:svd}: converting the matrix $A$ into a bi-diagonal form, and using a variant of the $QR$ algorithm to iteratively diagonalize this bi-diagonal matrix.

\section{Proof of Theorem~\ref{thm:np}}
\label{appendix1}
First, the decision version of our problem is in NP: Given a communication structure, computing the energy consumed at each node 
and checking if the constraints specified in Eqns (\ref{eqn:c1})-(\ref{eqn:c9}) are satisfied can both be done in polynomial time. Hence, testing feasibility
and whether the total cost is less than a given value $M$ is accomplished in polynomial time. 

Next, to prove NP-hardness we perform a reduction from the set cover problem~\cite{np:book}, whose decision version is defined as follows. 
\begin{definition}
Given a collection $C$ of subsets of a finite set $P$ and an integer $0 < K \leq |C|$, with $|C|$ the cardinality of $C$, 
the {\em set cover problem} asks whether $C$ contain a subset of $C' \subset C$ with $|C'| \leq K$, such that every element of $P$ belongs to at least
one of the subsets in $C'$ (this is called a {\em set cover} of $P$). 
\end{definition}

For any instance of the set cover problem, we now build an instance of the decision version of problem P1, which  
seeks to find a communication structure with which the energy cost is at most $M$ while satisfying the computational delay constraints
at each node and the combinability constraint. 

Consider a graph consisting of three layers, as shown in Figure~\ref{fig:np_proof}:  
a single node $V_0$ at the bottom layer, a (C) layer of sets of nodes $C_k \in C$ each with an internal structure as shown in Figure~\ref{fig:np_proof}(b), 
and a (P) layer of nodes $\{ p_j \in P \}$.  
Each element $C_k \in C$ in the middle layer contains the $|C_k|$ nodes in $C_k$ plus 3 extra nodes as shown in Figure~\ref{fig:np_proof}(b): 
Node $x_3$ connects to the base station $V_0$ with 0 weight. Nodes $x_1$ and $x_2$
are connected to $x_3$ with weights $1$ and $1 < a < d$, respectively, and connected with each other with weight $d$.  
The other $|C_k|$ nodes are connected to both $x_1$ and $x_2$ with weights $d > 0$. 
Furthermore, each such structure $C_k \in C$ are connected to the same $|C_k|$ nodes in the P layer that belong to the set $C_k$, via node $x_1$, all with weight $d$. 
Finally, all the $x_3$ nodes are inter-connected with weight $0$. 

\begin{figure}[htb]
 \centering
\includegraphics[width=3.5in]{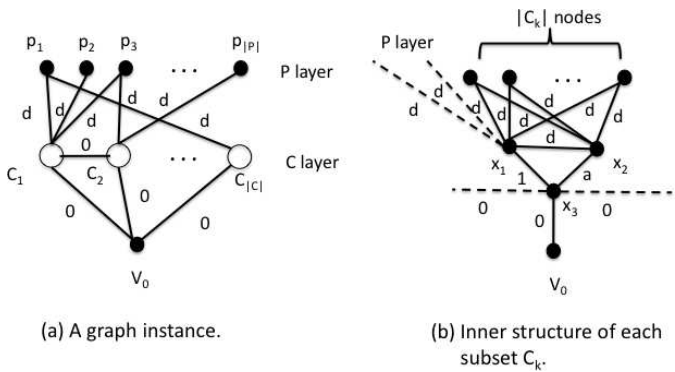}
  \caption{Instance of the problem P1 for any given instance of the set cover problem. In (b), the solid lines illustrate connectivity internal to the structure whereas dashed lines are for external connections. }
\vspace{-0.1in}
\label{fig:np_proof}
\end{figure}

The delay constraints are defined as follows: no more than $|C_k| + 1$ FFT's can be combined on nodes $x_1$ and $x_2$, and no more than $4$ FFT's on node $x_3$ for a given $C_k$, and no constraint on 
$V_0$ or any nodes in the $P$ layer.  

We are now ready to show that finding the solution to the decision version of P1 on the above network graph with the stated delay constraints and a choice of 
\begin{eqnarray}
M &=& d R \left( |P| + \sum_k  |C_k|  \right) + R |C| + a R |C| \nonumber \\
&& + a r \left( |P| + K \right) + r \sum_k \left( |C_k| + 1 \right) ~, 
\end{eqnarray}
for some positive integer $K \leq |C|$, is equivalent to finding a set cover of cardinality $K$ or less for the set $P$. 

For $d > \left( R|C| + a R |C| + a r \left( |P| + K \right) + r \sum_k \left( |C_k| + 1 \right) \right)/R$, the communication structure
for P1 will have transmissions on exactly $|P|$ edges between the layers $P$ and $C$, 
and on exactly $|C_k|$ edges in the structure shown in Figure~\ref{fig:np_proof}(b) for every $C_k \in C$. 
That means no other node than $x_1, x_2$ and $x_3$ will be used as a relay, or belong to $S$. If some other node belongs to $S$, then the cost of the communication structure
would contain $R$ bits passing through more than $|P| + \sum_k \left( |C_k| \right)$ edges of weight $d$ which would result in a cost larger than $M$. 
This also implies that $x_1$ and $x_3$ for all $C_k \in C$ belong to $S$. The only degree of freedom is whether $x_2$ lies in $S$ or not. 
(Recall that $x_2 \in S$ only if a SVD computation takes place on $x_2$ also.) 

The key idea is to show that for $1 < a < d$, finding a communication structure with cost at most $M$ means connecting the nodes in layer $P$
to at most $K$ structures of layer $C$. If more than $K$ structures in $C$ is needed, then the cost of the communication
structure will necessarily be higher than $M$. 

We first show that if a corresponding $C_k$ is not connected to any node in the $P$ layer, then its corresponding $x_2$ node will not belong to $S$.
This is because in this case the optimal communication structure is to have all the other $|C_k|$ nodes (other than $x_1, x_2$ and $x_3$) send their 
data to $x_1$ (since $a > 1$, transmitting everything to $x_1$ instead of $x_2$ will consume less energy) 
who will then compute the SVD and send the corresponding eigenvectors as well as its own FFT to $x_3$. Note that node $x_3$ can receive 
FFT's from $x_1, x_2$ and other $x_3$ nodes 
with no extra cost; thus combinability will be trivially ensured. 
It will forward all the computed eigenvectors to $V_0$. The total energy consumed in this operation is 
\begin{eqnarray}
E^1_k = d |C_k| R + a R + R + r \left( |C_k| + 1 \right). 
\end{eqnarray}

We next show that if a corresponding $C_k$ is connected to at least one node in the $P$ layer, then its corresponding $x_2$ will always belong to $S$. 
This is because since no more than $|C_k| + 1$ FFT's can be combined on $x_1$, if $n_k$ of the nodes in the $P$ layer send their FFT to $x_1$, then 
$x_1$ can combine FFT's from no more than $|C_k| - n_k$ nodes belonging to the structure of $C_k$. The remaining $n_k$ nodes 
will have to send their FFT to $x_2$ as it has the next smallest distance (after $x_1$) to these nodes. Thus, $x_2$ will combine data from $n_k+1$ nodes. 
The energy consumed in this scenario is 
\begin{eqnarray}
E^2_k = d |C_k| R + d n_k R + R + a R + a r (n_k+1) + r \left( |C_k| + 1 \right).
\end{eqnarray} 
Note that $\sum_{C_k} n_k = P$.

Thus, if the structures in the $P$ layer connected to the $C$ layer constitute the set $F_1$ and those unconnected to the $C$ layer the set $F_2$, 
then the total energy consumed is 
\begin{eqnarray}
E &=& \sum_{C_k\in F_1} E^1_k + \sum_{C_k \in F_2} E^2_k \nonumber \\
&=& d R \left( |P| + \sum_k  |C_k|  \right) + R |C| + a R |C| \nonumber \\
&& + a r \left( |P| + K' \right) + r \sum_k \left( |C_k| + 1 \right) ~,  
\end{eqnarray}
where $K' = |F_1|$.  The above quantity will be larger than $M$ if $K' > K$. This means that finding a communication structure with a cost at most $M$ implies 
finding a set of $K$ elements or less from the $C$ layer to which all the nodes in set $P$ connect. In other words, a communication
structure with a cost of at most $M$ yields a set cover of size at most $K$. 

Lastly, we need to ensure that the set cover of size at most $K$ also yields a communication structure of cost at most $M$. 
If an element is contained in only one set in the set cover, then connect the corresponding node in $P$ to the corresponding $C_k$, 
and if an element belongs to multiple sets in the set cover, then choosing one of these sets uniformly at random,
and connecting the corresponding node $P$ to the $C_k$ which corresponds to this randomly chosen set, 
yields an energy cost of no more than $M$ (follows obviously from the previous discussion). 
The computational delay constraint is also obviously satisfied at all nodes. We merely need to ensure that all computations are combinable. 
Since each node $x_3$ belonging to the structure of $C_k$ send its FFT to the node $x_3$ belonging to the structure of node $C_{\left(k+1\right) \mbox{mod} |C|}$,
all computations are combinable. 

Thus our decision problem is NP-complete and our optimization problem is NP-hard.

\section{Parallel Simulated Annealing}
\label{sec:annealing}
In a structural health monitoring system, a common technique to translate raw sensor data into an estimate of damage involves comparing system
properties in an unknown state of health to those in a known, undamaged state~\cite{model_update1,model_update2}. This technique is referred to as model updating
and involves adjusting the system parameters iteratively in an analytical model such that the analytical system produces response data that
matches results obtained experimentally. Using this method, damage can be detected in a system by periodically searching for changes
in model parameters that can be linked directly to suboptimal system performance. 

A wide variety of model updating techniques have been developed over the years~\cite{model_update3}. One common approach is to define an objective function, $E$, 
which relates the difference between analytical and experimental data. This function can be repeatedly evaluated with varying values of the analytical
model parameters until the difference between the analytical and experimental response is minimized. 

Simulated annealing (SA) is one of the most common algorithms for stochastically searching for the global minimum of such an objective function.
This method has been used frequently in model-based damage detection techniques~\cite{sa1}. Metropolis \etal~\cite{sa2} developed this algorithm 
to determine the global minimum energy state amidst a nearly infinite number of possible configurations. The Metropolis criterion expresses the
probability of a new system state being accepted at a given system temperature, and can be stated as: accept the new state if and only if 
$E_{new} \leq E_{old} - T ln \left( U \right)$, where $E$ is the value of the objective function for a given energy state, $U$ is a
uniformly distributed random variable between $0$ and $1$, and $T$ is the temperature of the system. The addition of the $T ln \left( U \right)$
term allows the system to accept an invalid state in the hope of avoiding premature convergence to a local minima. 

A standard SA algorithm begins the optimization process by assigning an initial temperature $T_1$, and letting the Metropolis algorithm run
for $N_1$ iterations. During each iteration, certain analytical model parameters are reassigned in a pseudo-random fashion, and the objective
difference between the experimental and analytical output is determined. This newly created state is either accepted or rejected based 
on the Metropolis criterion. After $N_1$ iterations, the temperature of the system is reduced to $T_2$ and the process runs for $N_2$ iterations.
This process continues till the temperature drops to a really low temperature, $T_M$, where very few new states are accepted, and the system has, in essence, frozen.
To summarize, the process runs for $M$ temperature steps, and for $N_j, 1 \leq j \leq M$ iterations for each temperature step $T_j$. 

Over the years, many parallel SA techniques have been developed and successfully implemented~\cite{sa3}.
Zimmerman \etal~\cite{andy:sa} proposed a new parallel SA technique more suited to be implemented over a wireless sensing system for structural health monitoring
as it reduces the communication required between processing nodes. This technique breaks up the traditionally serial SA tree (which is continuous
across all temperature steps) into a set of smaller search trees, each of which corresponds to a given temperature step and begins with the global
minimum values for the preceding temperature step. Each of these smaller trees can be assigned to a cluster of available nodes in the network, 
and thus can run concurrently. 

As the parallelized search progresses, updated global state information has to be disseminated downwards (to the nodes doing the computations at lower
temperatures) through the network. Specifically, when a node detects a new global minimum energy state at a given temperature, it communicates this
information to the cluster-head of its cluster, which then propagates this information to all nodes doing the computation at lower temperatures.
These nodes (computing at lower temperatures) will re-start their search based on this new state. This may seem wasteful at high temperatures, 
however, as the search algorithm converges on a solution, it becomes decreasingly likely that a new global minimum will be found at a given
temperature step which reduces the total number of transmissions. 
 
\cite{andy:sa} explores the advantages of this approach in a wireless sensing system. However, it does not explore how to construct the 
communication structure so as to minimize the energy consumption which will be the focus of this section. 

We now precisely state the problem. 
The designer will set the values of $N_j$ and $k_j, 1 \leq j \leq M$, which denote the number of computations
to be performed at temperature $T_j$ and the number of sensor nodes performing the computation at $T_j$ respectively. 
(Note that $\sum_{j=1}^M k_j \leq |V|$.) The values of $N_j$ and $k_j$ will be determined based on the accuracy and the computational constraint per node. 

Given the values of $N_j$ and $k_j$, determining the communication structure involves dividing the $V$ nodes into $M$ clusters
each of size $k_j, 1 \leq j \leq M$ and choosing a cluster-head for each cluster. 
Let the cluster of nodes corresponding to temperature $T_j$ be denoted by $K_j$. (Note that $|K_j| = k_j$.) 
Finally, let $b_j \in K_j$ denote the cluster-head for the cluster $K_j$. 
Any computation which results in a new minimum energy state at a temperature $T_j$ requires exchanging this information between all 
nodes belonging to the cluster $K_j$, between the cluster-heads $b_j$ and $b_l, l > j$, and all nodes belonging to clusters
$K_l, l > j$. Thus, the total number of transmissions for each new minimum energy state found at temperature $T_j$ is equal to
$\sum_{v \in K_j} H_{b_j \rightarrow v} + \sum_{l=j+1}^M H_{b_j \rightarrow b_l} + \sum_{l=j+1}^M \sum_{v \in K_l} H_{b_l \rightarrow v}$,
where recall that $H_{i \rightarrow j}, i,j \in V$ denotes the average number of transmissions required to exchange information between nodes $i$ and $j$
along the shortest path between the two nodes. 

We first describe an ILP to determine the optimal communication structure for parallel simulated annealing. 
Let $x_{ij}, i \in V, 1 \leq j \leq M$ be an indicator variable which is set to $1$ only if node $i \in K_j$. 
Let $y_{ij}, i \in V, 1 \leq j \leq M$ be another indicator variable which is set to $1$ only if node $i = b_j$, that is,
$i$ is the cluster-head for $K_j$. 
Note that here we have a separate variable to denote the cluster-head whereas for the SVD computation, we merely set $x_{ii}$ to $1$
if node $i$ was a cluster-head. The extra variable is needed for parallel simulated annealing to convert the quadratic objective into a linear equation. 
Let $t_{ikj}, i,k \in V, 1 \leq j \leq M$ denote an indicator variable which is set to $1$ only if node $i$ is the cluster-head for temperature $T_j$ and node $k \in K_j$
(that is $t_{ikj} = y_{ij} x_{kj}$) and let $p_{ikj}, i,k \in V, 1 \leq j \leq M-1$ denote an indicator variable which is set to $1$ only if node $i$ is the cluster-head
at temperature $T_j$ and node $k$ is the cluster-head at temperature $T_{j+1}$ (that is $p_{ikj} = y_{ij} y_{k(j+1)}$). 
Finally, let $a_j, 1 \leq j \leq M$ denote the probability of generating a new minimum energy state per computation at temperature $T_j$. 
Then, for $N_j$ computations at that temperature, the number of new minimum energy states generated are $a_j N_j$. 
Note that generating a new minimum energy state triggers new transmissions.   

Following is the ILP to determine the optimal communication structure for parallel simulated annealing.
\begin{eqnarray}
\label{sim:obj} & & \sum_{j=1}^M a_j N_j \left( \sum_{i \in V} \sum_{k \in V} H_{i \rightarrow k} \left( t_{ikj} + 
\sum_{l=j+1}^M p_{ikl} + \sum_{l=j+1}^M t_{ikl} \right) \right) \\
\label{sim:eqn1} & & \sum_{i \in V} x_{ij} = k_j, 1 \leq j \leq M \\
\label{sim:eqn2} & & t_{ikj} \geq \frac{y_{ij} + x_{kj} - 1}{2}, i,k \in V, 1 \leq j \leq M  \\
\label{sim:eqn3} & & p_{ikj} \geq \frac{y_{ij} + y_{k(j+1)} - 1}{2}, i,k \in V, 1 \leq j \leq M-1  \\
& & x_{ij}, y_{ij}, t_{ikj}, p_{ikj} \in \{ 0,1 \}, i,k \in V, 1 \leq j \leq M.  
\end{eqnarray}
The first constraint (Equation (\ref{sim:eqn1})) ensures that the cluster performing computations at temperature $T_j$ has $k_j$ nodes  
while the next two constraints populate the values of $t_{ikj} = y_{ij} x_{kj}$ and $p_{ikj} = y_{ij} y_{k(j+1)}$. 

We finally describe a greedy approximation algorithm to determine the communication structure for parallel simulated annealing.
Recall that we need to determine the set of nodes which form a cluster as well the corresponding cluster-head for each temperature $T_j, 1 \leq j \leq M$.
Figure~\ref{algo:sim_anneal_approx} describes the greedy algorithm. We first start from the smallest temperature $T_M$ because finding a new
energy state at any temperature will trigger a transmission between the cluster-head $b_M$ and the nodes belonging to the cluster $K_M$. 
Amongst all the nodes $v \in V$, determine the cluster-head to be the node which has the smallest sum of the average number of transmissions
required to get to $k_M$ nodes. This yields both $b_M$ and $K_M$. From amongst the remaining nodes, in a similar manner, greedily select $b_{M-1}$ 
and $K_{M-1}$ and continue. Assuming the maximum height of the unconstrained data collection tree (defined in Definition~\ref{def:dct}) 
is $O\left( log (|V|)\right)$, using arguments similar to ones made in Section~\ref{sec:approx}, 
the approximation factor of the greedy approximation algorithm is also $O\left( log (|V|)\right)$.

\begin{figure}[t]
  \begin{center}
  \begin{ttfamily}
    \begin{footnotesize}
      \flushleft
       $K = V$, $j=M$ \\
       while $(j>0)$ do \\
       \hspace*{0.2in} $minE = \infty$ \\ 
       \hspace*{0.2in} $\forall v \in K$ \\ 
       \hspace*{0.4in} $\left( e_v^j, k_v^j \right)$ = findMin $\left(K, v, j\right)$ \\ 
       \hspace*{0.4in} If $\left( MinE > e_v^j \right)$ \\
       \hspace*{0.6in} $MinE = e_v^j$, $K_j = k_v^j$, $b_j = v$ \\
       \hspace*{0.2in} $j=j-1$, $K = K \backslash K_j$ \\
       $ $ \\ 
       $ $ \\
       findMin $\left(K, b, j\right)$ \\
       \hspace*{0.2in} $T = \phi$ \\
       \hspace*{0.2in} $\forall v \in S$ \\
       \hspace*{0.4in} $d_v = H_{v \rightarrow b}$\\ 
       \hspace*{0.2in} Sort the nodes in $K$ in ascending order of $d_v$'s \\ 
       \hspace*{0.2in} Add the first $k_j-1$ nodes from this sorted list \\
        \hspace*{0.2in} to $T$ \\
       \hspace*{0.2in} $E = \sum_{v \in T} H_{v \rightarrow b} + I_{j<M} H_{b_{j+1} \rightarrow b}$ \\
       \hspace*{0.2in} ($I_{j<M}$ is an indicator variable which is equal to \\
       \hspace*{0.2in} $1$ if $j<M$, else it is equal to $0$.) \\
       \hspace*{0.2in} return $\left( E, T \right)$ \\
       \end{footnotesize}
\end{ttfamily}
\end{center}
  \caption{A greedy approximation algorithm to determine the communication structure for parallel simulated annealing.}
  \label{algo:sim_anneal_approx}
\end{figure}




\ifCLASSOPTIONcaptionsoff
  \newpage
\fi

\end{document}